\documentclass[lettersize,journal]{IEEEtran}

\usepackage{amsmath,epsfig,amssymb,verbatim,amsopn,cite,multirow}
\usepackage{amsthm}
\usepackage{balance}
\usepackage{multirow}
\usepackage[usenames,dvipsnames]{color}
\usepackage[all]{xy}  %%% used to make block diagram
\usepackage{url}
\usepackage{amsfonts}
\usepackage{amssymb}
\usepackage{epsfig}
\usepackage{epstopdf}
\usepackage{bm}
\usepackage{graphicx}
\usepackage{subcaption}
\usepackage{footnote}
\usepackage{tabularx}
\usepackage{threeparttablex}\usepackage{array}

\usepackage{algorithm,algorithmic}
\usepackage{amssymb}
\usepackage{hyperref}

  {\proof}{\proofend}
\newtheorem{proposition}{Proposition}

\newcommand{\qa}{{\bf a}}

\newcommand{\qc}{{\bf c}}

\newcommand{\qe}{{\bf e}}

\newcommand{\qg}{{\bf g}}

\newcommand{\qt}{{\bf t}}

\newcommand{\qv}{{\bf v}}
\newcommand{\qw}{{\bf w}}
\newcommand{\qx}{{\bf x}}
\newcommand{\qy}{{\bf y}}
\newcommand{\qz}{{\bf z}}

\newcommand{\qB}{{\bf B}}

\newcommand{\qF}{{\bf F}}
\newcommand{\qG}{{\bf G}}

\newcommand{\qI}{{\bf I}}

\DeclareMathOperator*{\argmax}{arg\,max}

\newcommand{\SINR}{\mathtt{SINR}}

\newcommand{\ZF}{\mathtt{ZF}}
\newcommand{\PZF}{\mathtt{PZF}}

\newcommand{\MR}{\mathtt{MR}}
\newcommand{\CS}{\mathtt{cs}}

\newcommand{\SI}{\mathtt{SI}}

\newcommand{\Ntx}{N}
\newcommand{\Nrx}{N}
\newcommand{\CL}{\mathtt{cl}}
\newcommand{\Ncl}{N_\CL}

\newcommand{\dl}{\mathtt{J}}
\newcommand{\ul}{\mathtt{O}}
\newcommand{\UT}{\mathtt{UT}}
\newcommand{\UR}{\mathtt{UR}}

\newcommand{\hgmkpd}{\hat{\qg}_{mk'}^{\dl}}
\newcommand{\tgmkd}{\tilde{\qg}_{mk}^{\dl}}
\newcommand{\tgmlu}{\tilde{\qg}_{mk}^{\ul}}

\newcommand{\gmlu}{{\qg}_{m\ell}^{\ul}}
\newcommand{\hGm}{\hat{\qG}_{\Wm}^{\ul}}

\newcommand{\gmkd}{\qg_{mk}^{\dl}}

\newcommand{\hgmkd}{\hat{\qg}_{mk}^{\dl}}

\newcommand{\hgnkd}{\hat{\qg}_{nk}^{\dl}}

\newcommand{\hgmku}{\hat{\qg}_{mk}^{\ul}}

\newcommand{\gmku}{\qg_{mk}^{\ul}}
\newcommand{\hgmlu}{\hat{\qg}_{m\ell}^{\ul}}

\newcommand{\vmku}{{\qv}_{mk}^{\ul}}
\newcommand{\vmkz}{{\qv}_{mk}^{\ZF}}
\newcommand{\vmkpzf}{{\qv}_{mk}^{\PZF}}
\newcommand{\vmkm}{{\qv}_{mk}^{\MR}}
\newcommand{\dm}{\delta_{m}}
\newcommand{\OT}{\mathcal{M}_k}
\newcommand{\SZ}{\mathcal{Z}_k}

\newcommand{\gamdmk}{\gamma_{mk}^{\dl}}
\newcommand{\gamdnk}{\gamma_{nk}^{\dl}}
\newcommand{\gamdmkp}{\gamma_{mk'}^{\dl}}
\newcommand{\gamumk}{\gamma_{mk}^{\ul}}

\newcommand{\gamuml}{\gamma_{m\ell}^{\ul}}

\newcommand{\Sm}{\mathcal{S}_m}
\newcommand{\Wm}{\mathcal{W}_m}
\newcommand{\Zk}{\mathcal{Z}_k}
\newcommand{\Tk}{\mathcal{M}_k}

\newcommand{\betamkd}{\beta_{mk}^{\dl}}

\newcommand{\betakkdu}{\beta_{kk}^{\mathtt{U}}}

\newcommand{\betakldu}{\beta_{\ell k}^{\mathtt{U}}}

\newcommand{\betamlu}{\beta_{m\ell}^{\ul}}
\newcommand{\betamku}{\beta_{mk}^{\ul}}

\newcommand{\alphmk}{\alpha_{mk}}

\DeclareMathOperator{\F}{\mathbf{F}} 

\DeclareMathOperator{\aaa}{\mathbf{a}}

\DeclareMathOperator{\K}{\mathcal{K}}

\DeclareMathOperator{\MM}{\mathcal{M}}

\DeclareMathOperator{\THeta}{\boldsymbol{\theta}}

\DeclareMathOperator{\ALPHA}{\boldsymbol{\alpha}}

\DeclareMathOperator{\diag}{\mathrm{diag}}

\title{\fontsize{0.83cm}{1cm}\selectfont   Cell-Free Massive MIMO Surveillance  of Multiple Untrusted Communication Links}
\author{Zahra Mobini,~\IEEEmembership{Member,~IEEE,} Hien Quoc Ngo,~\IEEEmembership{Senior Member,~IEEE,}\\ Michail Matthaiou,~\IEEEmembership{Fellow,~IEEE}, and Lajos Hanzo,~\IEEEmembership{Life Fellow,~IEEE}}

\allowdisplaybreaks
\allowdisplaybreaks
\begin{document}

\bstctlcite{IEEEexample:BSTcontrol}
\maketitle

\begin{abstract}
A cell-free massive multiple-input multiple-output (CF-mMIMO) system is considered   for enhancing the monitoring performance of wireless  surveillance, where a large number of distributed multi-antenna aided legitimate monitoring nodes (MNs)  proactively monitor multiple distributed untrusted communication links. We consider two types of MNs whose task is to either observe the untrusted transmitters or jam the untrusted receivers.  We  first analyze the performance of CF-mMIMO  surveillance relying on both maximum ratio (MR) and partial zero-forcing (PZF)  combining  schemes and derive closed-form expressions for the monitoring success probability (MSP) of the MNs. We then  propose a joint optimization technique that designs the MN mode assignment, power control, and MN-weighting coefficient control  to enhance the MSP based on the long-term statistical channel state information  knowledge. This challenging problem is effectively transformed into tractable forms and  efficient algorithms are proposed for solving them. Numerical results show that our proposed CF-mMIMO surveillance system  considerably improves  the monitoring performance with respect to a full-duplex co-located massive MIMO proactive monitoring system. More particularly, when the untrusted pairs are distributed over a wide area and use the MR combining, the proposed solution provides nearly a thirty-fold improvement in the minimum MSP over the co-located massive MIMO baseline, and forty-fold improvement, when the PZF combining is employed.

\let\thefootnote\relax\footnotetext{

This work is a contribution by Project REASON, a UK Government funded project under the Future Open Networks Research Challenge (FONRC) sponsored by the Department of Science Innovation and Technology (DSIT). It was also supported by the U.K. Engineering and Physical Sciences Research
Council (EPSRC) (grant No. EP/X04047X/1). The work of Z.~Mobini and  H.~Q.~Ngo
 was supported by the U.K. Research and Innovation Future
Leaders Fellowships under Grant MR/X010635/1, and a research grant from the Department for the Economy Northern Ireland under the US-Ireland R\&D Partnership Programme. The work of M. Matthaiou was supported by the European
Research Council (ERC) under the European Union’s Horizon 2020 research
and innovation programme (grant agreement No. 101001331).
L. Hanzo would like to acknowledge the financial support of the Engineering and Physical Sciences Research Council projects EP/W016605/1, EP/X01228X/1, EP/Y026721/1 and EP/W032635/1 as well as of the European Research Council's Advanced Fellow Grant QuantCom (Grant No. 789028).

Z. Mobini, H. Q. Ngo, and M. Matthaiou  are with the Centre for Wireless Innovation (CWI), Queen's University Belfast, BT3 9DT Belfast, U.K. Email:\{zahra.mobini, hien.ngo, m.matthaiou\}@qub.ac.uk.

L. Hanzo is with the School of Electronics and Computer
Science, University of Southampton, SO17 1BJ Southampton, U.K. (e-mail:
lh@ecs.soton.ac.uk).

Parts of this paper were presented at the 2023 IEEE GLOBECOM conference~\cite{zahra:2023:Globecom}.}
\end{abstract}
%==============================================================================

%==============================================================================
\section{Introduction}
%===========================================================================\text
The widespread use of mobile devices along with the explosive popularity of  wireless data services offered by fifth generation (5G) networks has led to the emergence of  the so-called infrastructure-free communication systems, which include  device-to-device (D2D), aerial vehicle (UAV)-aided communications, internet and so on. Although these wireless transmission systems provide an efficient and convenient means for establishing direct connections between mobile terminals, unauthorised or malicious users may misuse these networks to perform illegal activities, commit cyber crime, and jeopardize public safety. As a remedy, legitimate monitoring has attracted considerable attention in recent years~\cite{Zhang:MAg:2018}.

In contrast to wireless physical-layer security (PLS), which aims for making the transmitted information indecipherable to illegitimate monitors~\cite{ZAHRA:TIFS:2019},  this line of PLS  puts emphasis on legally monitoring  the communications of an untrusted pair. Wireless surveillance  is typically classified into three main paradigms based on the kind of strategies used in the surveillance process by a legitimate monitor: $1)$ passive monitoring~\cite{Zhang:TWC:2017}, $2)$ proactive monitoring~\cite{Zhang:TWC:2017}, and $3)$ spoofing relaying~\cite{Jiang:LSP:2017}. In passive monitoring, the legitimate monitor silently observes an untrusted link and, hence, successful monitoring  can only be achieved for the scenarios where  the strength of the monitoring link is better than that of the untrusted one. 
In proactive monitoring, the legitimate monitor operates in a full-duplex (FD) mode, simultaneously observing the untrusted link and sending a jamming signal to interfere with the reception of the untrusted receiver (UR), thereby  degrading the rate of the untrusted link. This improves the    monitoring success probability (MSP)~\cite{Zhong:TWC:2017}, which is one of the fundamental monitoring performance objectives. Finally, in the context of spoofing relaying, the FD monitor observes the untrusted link  and alters the channel information sent over the untrusted link to adjust its rate requirement.

%-----------------
\begin{table*}
\caption{Explicitly Contrasting Our Contributions to the Literature}
\vspace{-0.5em}
\centering
\begin{tabular}{||c||c|c|c|c|c|c|c|c|c||}
	\hline
	Feature &  \cite{Zhang:TWC:2023} \  &  ~\cite{Xu:TWC:2022} & \cite{Li:LSP:2018} & \cite{Haiyang:TWC:2020} &  \cite{Xu:TVT:2020} & \cite{Baogang:TVT:2019} & \cite{Moon:TWC:2019} & \cite{Li:TIFS:2021}& \textbf{our work}\\
	\hline\hline
	Multiple untrusted pairs &   &  \checkmark &\checkmark  & \checkmark$^a$ & \checkmark  &\checkmark  &  & &\checkmark\\
	\hline
	Multiple monitors &   &   &  &  &   & \checkmark$^b$ & \checkmark &\checkmark$^b$ &\checkmark\\
	\hline
	Multiple-antenna technique & \checkmark  &   &  &\checkmark   &   &  &  & &\checkmark\\
	\hline
	Statistical CSI & $\checkmark$  &   &  &  &   &  &  & &\checkmark\\
	\hline
	 Coherent monitoring&   &   &  &  &   &  &  & &\checkmark\\
 	\hline
	Distributed  &   &   &  &  &   &  &  & &\checkmark\\
 	\hline
  \end{tabular}
  \begin{tablenotes}
   \item \hspace{9em}$^a$ Two untrusted pairs, $^b$ two monitors. 
\end{tablenotes}
\vspace{0.3em}
\label{t1}
\end{table*}
\setlength{\textfloatsep}{0.4cm}
%----------------------------------

Since the publication of the pioneering paper~\cite{Zhang:TWC:2017}, proactive monitoring has been widely studied  under diverse
untrusted communication scenarios, such as  multiple-input multiple-output (MIMO) systems~\cite{Zhong:TWC:2017,Feizi:TCOM:2020,Huang:LSP:2018,Zhang:TWC:2023}, relaying systems~\cite{Jiang:LSP:2017,Dingkun:LSP:2017, Moon:2018:TWC,Li:2018:GLOBECOM}, UAV networks~\cite{Mobini:ICC:2018,Kai:TVT:2019},  cognitive radio networks~\cite{Xu:TVT:2021,Ge:ITJ:2022}, and intelligent reflecting surface (IRS)-aided surveillance systems~\cite{Guojie:TVT:2022,Ming:TWC:2023}. Specifically, in~\cite{Zhong:TWC:2017,Feizi:TCOM:2020}, multi-antenna techniques were utilized to improve the monitoring performance. In~\cite{Zhong:TWC:2017}, an optimization framework  was developed for the jamming power control and transmit/receive beamforming vectors at the legitimate monitor  by maximizing the MSP. The authors of~\cite{Feizi:TCOM:2020} established an optimization framework for joint precoding design and jamming power control, by taking into account the impact  of jamming on the performance of  other  legitimate users. Low-complexity suboptimal zero-forcing (ZF)-based beamforming schemes were also proposed in~\cite{Feizi:TCOM:2020}. 
Under the more practical assumption of imperfect channel state information (CSI), the maximization of the worst-case MSP attained by multi-antenna aided proactive monitoring systems was investigated in~\cite{Huang:LSP:2018}, where the CSI error was deterministically bounded. Additionally, by assuming the knowledge of the imperfect instantaneous CSI of the observing link and CSI statistics of the
jamming  and the untrusted links,   Zhang  \emph{et al.}~\cite{Zhang:TWC:2023} studied the performance of multi-antenna assisted proactive monitoring  in uplink systems and derived semi-closed-form expressions for both the MSP and monitoring rate.
Proactive monitoring was investigated in~\cite{Jiang:LSP:2017} in dual-hop decode-and-forward relaying systems, where the legitimate monitor can adaptively act  as a monitor, a jammer or a helper, while proactive monitoring  was studied in~\cite{Dingkun:LSP:2017} via jamming designed for amplify-and-forward relay networks. Moon \emph{et al.}~\cite{Moon:2018:TWC} extended the results of~\cite{Dingkun:LSP:2017} to  multi-antenna aided multi-relay systems.
Later, UAV-aided information surveillance was proposed in~\cite{Mobini:ICC:2018}, where an FD ground monitor observes the untrusted link  and simultaneously  sends the collected untrusted information to the UAV. By contrast, Li \emph{et al.}~\cite{Kai:TVT:2019}  relied on a legitimate UAV  to track  untrusted UAV-to-UAV communications and developed both an energy-efficient jamming strategy and  a tracking algorithm.
The proactive monitoring concept of
cognitive systems was introduced  in~\cite{Xu:TVT:2021} and~\cite{Ge:ITJ:2022}, where the secondary users are allocated to share the spectrum of the untrusted users, provided that they are willing to  act as an observer or friendly jammer monitoring the untrusted link.
Recently, IRSs have also found their way into information surveillance systems, where the IRS is used for degrading the untrusted channel's rate~\cite{Guojie:TVT:2022} and for improving  the observing channel~\cite{Ming:TWC:2023} to further enhance the monitoring performance. Finally, beneficial IRS
deployment strategies and joint beamforming design problems were proposed in~\cite{Ming:TWC:2023}.

\subsection{Knowledge Gap and Motivations}
It is important to point out that most  studies tend to investigate simple setups concerning the untrusted communication links and/or observing links. More specifically, a popular assumption in the aforementioned literature is that there is a single untrusted link.  This  assumption  is optimistic, because realistic systems are likely to have more than one untrusted communication links in practice. In this context, Xu and  Zhu~\cite{Xu:TWC:2022}  have studied proactive monitoring  using a single monitor for observing multiple untrusted pairs in  scenarios associated with either average rate or with outage probability constrained untrusted links. Li \emph{et al.} \cite{Li:LSP:2018} used   proactive monitoring with relaying features to increase the signal-to-interference-plus-noise ratio   (SINR) of multiple untrusted links,  which results in a higher rate for the untrusted links, and hence, higher observation rate.  Moreover, Zhang \emph{et al.} in~\cite{Haiyang:TWC:2020} characterized the achievable monitoring rate
region of  a single-monitor surveillance system observing two untrusted pairs operating within the same spectral band and using a minimum-mean-squared-error  successive interference cancellation (MMSE-SIC) receiver. Proactive monitoring was studied in~\cite{Xu:TVT:2020} for the downlink of an untrusted non-orthogonal multiple access (NOMA) network with one untrusted transmitter (UT) and multiple groups of URs, while relying on a single-antenna monitor equipped with a SIC receiver.

In the case of a distributed deployment of untrusted pairs over a geographically wide area, it is impractical to cater for the direct monitoring of each and every untrusted pair by relying on a single monitor. Hence, attaining a given target MSP performance for the untrusted pairs is a fundamental challenge.
Therefore,  cooperative operation relying on a single primary FD monitor and an auxiliary assistant FD monitor supervising a single UT and multiple URs  was proposed in~\cite{Baogang:TVT:2019} for maximizing the monitoring energy efficiency via optimizing the jamming power and the cooperation strategy selected from a set of four specific  strategies. Later, Moon \emph{et al.}~\cite{Moon:TWC:2019} looked into  proactive monitoring relying on a group of single-antenna aided intermediate relay nodes harnessed  for supporting a  legitimate monitor, which acts either as a jammer or as an observer node.   Furthermore, the authors of~\cite{Li:TIFS:2021} harnessed  a pair of single-antenna half-duplex nodes that take turns in performing observing and jamming. 
However, the significant drawback of these studies is that they only focused on either the single-untrusted-link scenario~\cite{ Moon:TWC:2019,Li:TIFS:2021} or on a specific system setup~\cite{Baogang:TVT:2019}. Another main concern  is the  overly optimistic assumption of knowing instantaneous  CSI of all links  at the  monitor nodes. In this case, the system level designs must be re-calculated on the small-scale fading time scale, which fluctuates quickly in both time and frequency.
Therefore, the study of how to efficiently carry out surveillance operation using {multiple monitors} in the presence of {multiple untrusted pairs} is extremely timely and important, yet, this is still an open problem at the time of writing.

To address  the need for reliable information surveillance in complex  practical scenarios, we are inspired by the emerging technique of cell-free massive MIMO (CF-mMIMO)~\cite{Hien:cellfree} to  propose a new proactive monitoring system, termed as \emph{CF-mMIMO surveillance}.
CF-mMIMO constitutes an upscaled version of user-centric network MIMO.
In contrast to traditional cellular
systems, \emph{i)} fixed cells and cell boundaries disappear in CF-mMIMO and \emph{ii)} the users
are served coherently by all serving antennas within the same time-frequency resources~\cite{Matthaiou:COMMag:2021}. Therefore, CF-mMIMO offers 
significantly higher degrees of freedom in managing interference, hence resulting in substantial performance improvements for all the users over conventional cellular networks.
The beneficial  features of CF-mMIMO are
its substantial macro diversity, favorable propagation, and ubiquitous coverage for all users in addition to excellent geographical load-balancing. Owing to these eminent advantages, CF-mMIMO has sparked considerable research interest in recent years and has yielded huge performance gains in
terms of spectral efficiency (SE)~\cite{Hien:cellfree, Bashar:TWC:2019}, energy efficiency~\cite{Hien:TGCN:2018,Zhang:IOT:2021},  and security~\cite{Tiep:TCOM:2018,Chen:IOT:2023}. 
More interestingly, recent research has shown that  utilizing
efficient power allocation and  receive combining/transmit precoding designs in CF-mMIMO, relying on multiple-antenna access points (APs) further enhances the system performance~\cite{Giovanni:TWC:2020, Du:TCOM:2021}.

 Our CF-mMIMO information
surveillance system is comprised of a large number of spatially
distributed legitimate multiple-antenna monitoring nodes (MNs), which jointly and
coherently perform surveillance of multiple untrusted pairs
distributed over a wide geographic  area. In our system,
there are typically several MNs in each other's close proximity
for any given untrusted pair. Thus, high macro-diversity gain
and low path loss can be achieved, enhancing
the observing channel rates and degrading the performance of untrusted links.
Therefore, CF-mMIMO surveillance is expected to offer an
improved and uniform monitoring performance for all the
untrusted pairs compared to its single-monitor (co-located)
massive MIMO based counterpart. In addition, the favorable propagation characteristics of CF-mMIMO systems allow our CF-mMIMO surveillance system to employ simple processing
techniques, such as linear precoding and combining, while
still delivering excellent monitoring performance\footnote{In general, network-wide signal processing
maximizes the system performance but it entails complex signal co-processing procedures, accompanied by substantial deployment costs. Hence, it is unscalable as the number of service antennas and/or users grows unboundedly. On the other hand, distributed
processing is of low-complexity and more scalable, but its performance is often far from the optimal one. However, for CF-mMIMO systems, as a benefit of the distributed network topology and massive MIMO properties, distributed signal processing can strike an excellent trade-off between the system performance and scalability.}. More importantly, when the CF-mMIMO concept becomes integrated
into our wireless surveillance system, a virtual FD mode can
be emulated, despite relying on half-duplex MNs. Relying on
half-duplex MNs rather than FD MNs, makes the monitoring
system more cost-effective and less sensitive to residual self interference. More particularly, two types of MNs
are considered: 1) a specifically selected subset of the MNs
is purely used for observing the UTs; 2) the rest of the MNs
cooperatively jam the URs. In addition, since the MNs are now
distributed across a large area, the inter-MN interference encountered is significantly reduced compared to a conventional
FD monitoring/jamming system. Moreover,  harnessing the
channel hardening attributes of CF-mMIMO systems enables
us to dynamically adjust the observing vs. jamming mode,
the MN transmit power, and the MN-weighting coefficients for
maximizing the overall monitoring performance based on only
long-term CSI. Table I boldly and explicitly contrasts our
contributions and benchmarks them against the state-of-the art. We further elaborate on the novel contributions of this
work in the next subsection in a point-wise fashion.

\subsection{Key Contributions}
The main technical contributions and key novelty of this paper are summarized as follows:
\begin{itemize}
\item We propose a novel wireless surveillance system, which is based on the CF-mMIMO concept relying on  either observing or jamming  mode assignment. In particular, by assuming realistic imperfect CSI knowledge,  we derive exact closed-form expressions for the MSP of CF-mMIMO surveillance system with multiple-antenna MNs over multiple untrusted pairs for distributed maximum ratio (MR) and partial ZF (PZF)  combining  schemes. Additionally, we show that when the number of MNs in observing mode tends to infinity, the effects of inter-untrusted user interference, inter-MN interference, and noise gradually disappear. Furthermore, when the number $M_\dl$ of MNs in the jamming mode goes to infinity, we can reduce the transmit power of each MN by a factor of $\frac{1}{M_\dl}$  while maintaining the given SINR.
    
    \item  We formulate a joint optimization problem for the MN mode
assignment, power control, and MN-weighting coefficient control  for  maximizing the minimum MSP  of all the untrusted pairs subject to a  per-MN average transmit power constraint.  We  solve the  minimum MSP maximization problem by casting the original problem into three sub-problems, which are solved using an iterative algorithm.

  \item We also  propose a  greedy  UT grouping algorithm for our CF-mMIMO system relying on PZF combining scheme.  Our numerical results show that the proposed  joint optimization approach significantly outperforms the random mode assignment, equal power allocation,  and equal MN-weighting  coefficient based approaches. 
 The simulation results also confirm that, compared to the co-located massive MIMO aided proactive monitoring system relying on FD operation, where  all MNs are co-located as an antenna array and simultaneously perform observation and jamming, our CF-mMIMO surveillance system brings the MNs geographically closer to the untrusted pairs. Thus leads to a uniformly good monitoring performance for all untrusted pairs\footnote{In terms of data-sharing overhead and fronthaul, co-located massive MIMO based surveillance systems require lower fronthaul capacity compared to CF-mMIMO surveillance. Nevertheless, in our CF-mMIMO surveillance system we consider local processing,  which strikes an excellent balance between the computational complexity, fronthaul limitations, and monitoring performance.}.
\end{itemize}

\textit{Notation:} We use bold upper case letters to denote matrices, and lower case letters to denote vectors.  The superscripts $(\cdot)^*$, $(\cdot)^T$, and $(\cdot)^\dag$ stand for the conjugate, transpose, and conjugate-transpose (Hermitian), respectively. A zero-mean circular symmetric complex Gaussian distribution having a variance of $\sigma^2$ is denoted by $\mathcal{CN}(0,\sigma^2)$, while $\mathbf{I}_\Ntx$ denotes the $\Ntx \times  \Ntx$ identity matrix.  Finally, $\mathbb{E}\{\cdot\}$ denotes the statistical expectation.

%==============================================================================
%==============================================================================
 %%====================================================

\begin{figure}[t]
	\centering
  \vspace{-10em}
	\includegraphics[width=0.5\textwidth]{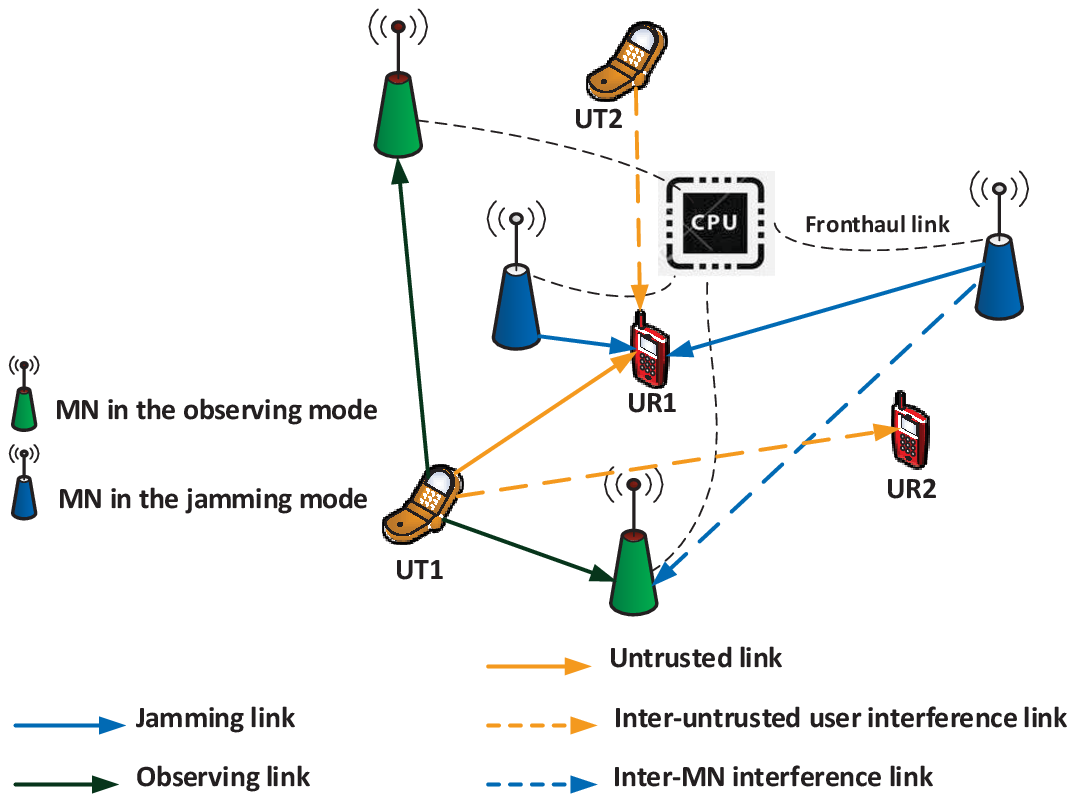}
	\vspace{-10em}
	\caption{CF-mMIMO surveillance system model with the assigned MNs in observing mode and jamming mode  along with the received desired and interference signals at a typical untrusted pair (UT $1$ and UR $1$).} 
	\vspace{-.4em}
	\label{fig:Fig0}
\end{figure}
%%======================================
  \vspace{-1em}
\section{System Model}~\label{sec:Sysmodel}
%==============================================================================
%==============================================================================
In this section, we introduce the CF-mMIMO surveillance system model for two different combining schemes. As shown in Fig.~\ref{fig:Fig0}, we consider a surveillance scenario, where $M$ MNs are employed to monitor  $K$ untrusted communication pairs. Let us denote the sets of MNs and untrusted communication pairs by $\MM \triangleq \{1, \dots, M\}$ and $\K\triangleq \{1,\dots,K\}$, respectively.
Each  UT and UR is equipped with a single antenna, while each MN is equipped with $\Ntx$  antennas.  All MNs, UTs,  and URs are half-duplex devices.  We assume that all MNs are connected to the central processing unit (CPU) via  fronthaul links.  The MNs can switch between observing mode, where they receive untrusted messages, and jamming mode, where they send jamming signals to the URs. The assignment of each mode to its corresponding MN is designed to maximize the minimum MSP over all the untrusted links, as it will be discussed in Section~\ref{sec:SE}. We use the binary variable $a_{m}$ to show the mode assignment for each MN $m$,  so that 
%-----------------------------------------
\vspace{-0.5em}
\begin{align}
\label{a}
a_{m} \triangleq
\begin{cases}
  1, & \text{if MN $m$ operates in the jamming mode,}\\
  0, & \mbox{if MN $m$ operates in the observing mode}.
\end{cases} 
\end{align}
%==============================================
Note that we consider block fading channels, where the  fading envelope of each link stays constant during the transmission of a block of symbols and changes to an independent value in the next block.
The jamming channel (observing channel) vector between the $m$-th MN and the $k$-th UR ($k$-th UT) is denoted by  $\gmkd\in\mathbb{C}^{\Ntx \times 1}$  ($\gmku\in\mathbb{C}^{\Nrx \times 1}$), $\forall k  \in  \K, m \in \MM$, respectively.
It is modelled as $\gmkd=\sqrt{\betamkd}\tgmkd,~(\gmku=\sqrt{\betamku}\tgmlu)$, where $\betamkd$ ($\betamku$) is the large-scale fading coefficient and $\tgmkd\in\mathbb{C}^{\Ntx \times 1}$ ($\tgmlu\in\mathbb{C}^{\Ntx \times 1}$) is the small-scale fading vector containing independent and identically distributed (i.i.d.) $\mathcal{CN} (0, 1)$ random variables (RVs). Furthermore, the channel gain between the $\ell$-th UT  and the $k$-th UR  is  $h_{\ell k}=(\betakldu)^{1/2}\breve{h}_{\ell k}$, where $\betakldu$ is the large-scale fading coefficient and $\breve{h}_{\ell k}$ represents small-scale fading, distributed as $\mathcal{CN}(0,1)$.
We note that $h_{kk}$ models the channel coefficient of the $k$-th
untrusted link spanning from the $k$-th UT to the $k$-th UR, $ \forall k \in \K$.
Finally, the channel matrix between MN $m$ and MN $i$, $\forall m,i\in\MM$, is denoted by $\qF_{mi}\in \mathbb{C}^{\Nrx\times\Ntx}$ where its elements, for $i\neq m$, are  i.i.d. $\mathcal{CN}(0,\beta_{mi})$ RVs and $\F_{mm} = \bold{0}, \forall m$.
Note that the channels $\gmkd$ and $\gmku$
may be estimated at the legitimate MN by overhearing the
pilot signals sent by UT $k$ and UR $k$, respectively~\cite{Moon:2018:TWC}. By following~\cite{Hien:cellfree}, for the minimum-mean-square-error (MMSE) estimation technique and the assumption of orthogonal pilot sequences, the   estimates of $\gmkd$  and $\gmku$ can be written as   $\hgmkd \sim \mathcal{CN}(\bold{0},\gamdmk \mathbf{I}_\Ntx)$ and $\hgmku \sim \mathcal{CN}(\bold{0},\gamumk \mathbf{I}_\Ntx)$, respectively, 
where $\gamdmk=\frac{\tau_\mathrm{t}\rho_\mathrm{t}(\betamkd)^2}{\tau_\mathrm{t}\rho_\mathrm{t}\betamkd+1}$ and $\gamumk=\frac{\tau_\mathrm{t}\rho_\mathrm{t}(\betamku)^2}{\tau_\mathrm{t}\rho_\mathrm{t}\betamku+1}$
with $\rho_\mathrm{t}$ and $\tau_\mathrm{t} \geq 2K$ being the normalized transmit power of each pilot symbol and the length  of pilot sequences, respectively.
Since it is difficult (if not impossible) for the legitimate MNs to obtain the CSI of untrusted links, we assume that  $h_{k\ell}$ is unknown to the  MNs.
%==============

All the UTs simultaneously send independent
untrusted messages to their corresponding URs over the
same frequency band. The signal  transmitted from  UT $k$ is denoted by $x_{k}^\ul  = \sqrt{\rho_\UT} s_{k}^{\ul}$, 
%--------------------
%--------------------
where $s_{k}^\ul$, with $\mathbb{E}\left\{|s_{k}^\ul|^2\right\}=1$, and $\rho_\UT$ represent  the transmitted symbol and  the  normalized transmit power at each UT, respectively. At the same time, the MNs in jamming mode intentionally send jamming signals to interrupt the communication links between untrusted pairs. This enforces the reduction of the achievable data rate at the URs, thereby enhancing the MSP. More specifically, the MNs operating in jamming mode use the MR transmission technique, also known as conjugate beamforming, in order  to jam the reception of the URs. Note that MR is considered because it maximizes the strength of the jamming signals at the URs.
Let us denote the jamming symbol intended for the untrusted link $k$  by $s_k^\dl$, which is a RV with zero mean and unit variance. When using MR precoding, the $\Ntx\times 1$  signal vector transmitted by MN $m$  can be expressed as
%----------------------------
\vspace{-0.3em}
\begin{align}
 \qx_{m}^{\dl}
= a_m\sqrt{\rho_\dl}\sum_{k \in \mathcal{K}} \sqrt{\theta_{mk}} \left(\hgmkd\right)^*
s_{k}^{\dl},   
\end{align}
%--------------------
where $\rho_\dl$ is the maximum normalized transmit power at each MN in the jamming mode. Moreover, $ \theta_{mk}$ denotes the  power allocation coefficient chosen to satisfy the practical power constraint $\mathbb{E}\left\{\|\qx_{m}^{\dl}\|^2\right\} \leq \rho_\dl$ at each MN in jamming mode, which can be further expressed as
%--------------------
\vspace{-0.2em}
\begin{align}
\label{DL:power:cons}
a_m\sum_{k\in\mathcal{K}} \gamdmk \theta_{mk} \leq \frac{1}{\Ntx}, \forall m.
% ,~\text{if}\,\, a_m = 1,
\end{align}
%--------------------
Accordingly, the signal received  by UR $k$ can be written as
%----------------------
\vspace{0.2em}
\begin{align}~\label{eq:ykdl}
&y_k^{\UR}
= h_{kk}x_{k}^{\ul}+ \sum_{\ell\in \mathcal{K}, \ell \neq k}h_{\ell k}x_{\ell}^{\ul}  
\nonumber\\
&+\sqrt{\rho_\dl}\!\!\sum_{m \in \mathcal{M}}
\sum_{k'\in\mathcal{K}}\!\! a_m\sqrt{\theta_{mk'}}
\left(\gmkd\right)^T\left(\hgmkpd\right)^*
s_{k'}^{\dl}+w_{k}^{\UR},
\end{align}
%----------------------
where $w_{k}^{\UR}\sim\mathcal{CN}(0,1)$ is the additive white Gaussian noise (AWGN) at UR $k$. It is notable that the second term in~\eqref{eq:ykdl} represents the interference caused by other UTs due to their concurrent transmissions over the same frequency band and the third term quantities the interference emanating from the MNs in the jamming mode.

%=================================================

%=================================================

The MNs in the observing mode, i.e., MNs  with $a_m=0, \forall m$, receive the transmit signals from all UTs. The received signal $\qy_{m}^{\ul}\in\mathbb{C}^{\Nrx \times 1}$ at MN $m$ in the observing mode is expressed as
\vspace{0.5em}
%----------------------
\vspace{-0.5em}
\begin{align}\label{eq:ymul}
\qy_{m}^{\ul}
=&
\sqrt{\rho_\UT}\sum_{k\in \mathcal{K}}(1-a_m)\qg_{mk}^{\ul} s_{k}^{\ul}
+\sqrt{\rho_\dl}
\sum_{i\in\mathcal{M}}\sum_{\ell\in \mathcal{K}}a_i\times\nonumber
\\
&
~~
(1-a_m) \sqrt{\theta_{i\ell}}
\qF_{mi}
(\hat{\qg}_{i\ell}^\dl)^*s_\ell^\dl
+(1-a_m)\qw_{m}^{\ul},
\end{align}
%----------------------------------
where $\qw_{m}^{\ul}$ is the $\mathcal{CN}(\bold{0}, \mathbf{I}_\Ntx)$ AWGN vector. We note  from~\eqref{eq:ymul}  that if MN $m$ does not operate in the observing mode, i.e., $a_m=1$, it does not receive any signal, i.e., $\qy_{m}^{\ul}=\boldsymbol{0}$. 
Then, MN $m$ in the observing mode performs linear combining by partially equalizing the received signal in~\eqref{eq:ymul} using the combining vector $\vmku$    as $(\vmku)^\dag\qy_{m}^{\ul}$.
The resultant signal is then forwarded to the CPU for detecting the untrusted signals, where the receiver combiner sums
up the equalized weighted signals. 
In particular, to enhance the observing capability, we assume that the forwarded signal is further multiplied by the MN-weighting  coefficient $\alphmk,  0 \leq\alphmk\leq 1,~\forall k, m$. The aggregated received signal for UT $k, \forall k,$ at the CPU is
%=========================
\begin{align}\label{eq:rul}
	r_{k}^{\ul}&=\sum_{m=1}^{M} \alphmk (\vmku)^\dag\qy_{m}^{\ul} \nonumber
 \\
&=  {\sqrt{\rho _\UT}\sum\limits_{m \in {\mathcal{M}}} \alphmk(1-a_m)  {(\vmku)^\dag}\gmku} s_k^\ul
+\nonumber
\\
&\sum\limits_{\ell \in {\mathcal{K}}\backslash k } {\sqrt{\rho _\UT}\sum\limits_{m \in {\mathcal{M}}}\alphmk (1-a_m) {(\vmku)^\dag}\gmlu }s_\ell^\ul
+\nonumber\\ 
& \sum\limits_{\ell \in {{\mathcal{K}}}}\sqrt{\rho _\dl}\!\! \sum\limits_{m \in {\mathcal{M}}} \!\sum\limits_{i \in {\mathcal{M}}}\!\!\alphmk(1-a_m)a_i\sqrt{\theta_{i\ell}}  
 {(\vmku)^\dag}\times
 \nonumber
 \\
 &{\qF_{mi}}{({\mathbf{\hat g}}_{i\ell}^\dl)^{\ast}}s_\ell^\dl
 \!+\sum\limits_{m \in {\mathcal{M}}}\!\! \alphmk(1-a_m) {(\vmku)^\dag} \qw_{m}^{\ul}.
\end{align}
%=======================
Finally, the observed information $s_{k}^{\ul}$ can be detected from $r_{k}^{\ul}$.
% %--------------------
%=================================================
\subsection{Combining Schemes}\label{sec:combining_scheme}
For CF-mMIMO surveillance systems, when multiple untrusted pairs are spatially multiplexed, the  linear receive combiner may harness the   MMSE objective function (OF), albeit  other OFs may also be harnessed. However,  MMSE optimization relies on centralized processing  and has high computational complexity as well as signaling
load. As a potential remedy,  low-complexity interference-agnostic combining schemes, such as MR, perform well, provided that each MN is equipped with a large number of antennas. But, the MR combiner
does not perform well for two scenarios: 1) when no favorable propagation can be guaranteed between the untrusted users, namely when the MNs are only equipped with a few antennas; and 2) in  interference-limited regimes, since MR is incapable of eliminating the inter-untrusted-user interference. In these cases, partial ZF-based  combining  outperforms MR due to its ability to deal with interference, while still being scalable.  Therefore, in this paper, we consider both the partial ZF and MR combining schemes, which can be implemented in a distributed manner and do not require any instantaneous CSI exchange between the MNs and the CPU.

 \emph{1) Maximum Ratio Combining}:
The simplest linear combining solution is the MR combining (i.e., matched filter) associated with 
%----
\begin{align}\label{eq:MR_PREC}
\vmku=\vmkm=\hgmku,
\end{align}
which has low computational complexity. MR combining  maximizes
the power of the desired observed signal, while retaining the system's scalability. In this case, we have
%====
\begin{align}\label{eq:rulMR}
	r_{k}^{\MR}&=\sum_{m=1}^{M} \alphmk (\hgmku)^\dag\qy_{m}^{\ul} \nonumber
 \\
&=  {\text{DS}_k^\MR} s_k^\ul+\sum\limits_{\ell \in {\mathcal{K}}\backslash k }{{\text{UI}}_{\ell k} ^\MR}s_\ell^\ul+\sum\limits_{\ell \in {{\mathcal{K}}}}{\text{MI}_{\ell k} ^{\MR}}s_\ell^\dl+{{\text{AN}}_k ^\MR}
\end{align}
%----------------
where
%---------
%====
\begin{align}
&{\text{DS}_k^\MR} =  {\sqrt{\rho _\UT}\sum\limits_{m \in {\mathcal{M}}} \alphmk(1-a_m)  {(\hgmku)^\dag}\gmku},\nonumber
\\
&{{\text{UI}}_{\ell k} ^\MR}= \sqrt{\rho _\UT}\sum\limits_{m \in {\mathcal{M}}}\alphmk (1-a_m) {(\hgmku)^\dag}\gmlu,
\nonumber
\\ 
&\text{MI}_{\ell k} ^{\MR}= \sqrt{\rho _\dl}\!\!\sum\limits_{m \in {\mathcal{M}}} \sum\limits_{i \in {\mathcal{M}}}\!\!\alphmk(1-a_m)a_i\sqrt{\theta_{i\ell}}  
 {(\hgmku)^\dag}{\qF_{mi}}{({\mathbf{\hat g}}_{i\ell}^\dl)^{\ast}},\nonumber
 \\
&{\text{AN}}_k ^\MR= \sum\limits_{m \in {\mathcal{M}}} \!\!\alphmk(1-a_m) {(\hgmku)^\dag}\qw_{m}^{\ul}, 
\end{align}
%----------------
%-------
where  $\text{DS}_k ^\MR,  \text{UI}_{\ell k} ^\MR$, and $\text{MI}_{\ell k} ^\MR$,    represent the desired
signal, cross-link interference caused by  the transmission of $\ell$-th UT, and
inter-MN interference,  respectively. Furthermore,  $\text{AN}_k ^\MR$ represents the additive noise.
%======

\emph{1) Partial Zero-Forcing Combining}:
The MR combining does not perform well at high signal-to-noise ratios (SNRs), since it  is incapable of eliminating the inter-untrusted user interference. For this reason, we now consider  the PZF combining scheme, which has the ability to mitigate interference in a  distributed and scalable manner while   attaining a flexible trade-off between the interference mitigation and  array gain~\cite{Jiayi:TWC:2021}. % In particular, we apply the  PZF combining that only cancel the interference generated by strong UTs that have strong channel gains, while the interference caused by weak UTs  is tolerated.
Therefore, each MN $m$ in observing mode virtually divides the UTs into two groups: $\Sm \subset \{1, 
\ldots, K\}$, which includes the index of strong UTs, and $\Wm \subset \{1, \ldots, K\}$, which hosts the index of weak UTs, respectively.  The  UT grouping can be based on diverse criteria, including the value of large-scale fading coefficient $\betamku$. Our proposed UT grouping strategy will be discussed in  section~\ref{Sec:numerical}. Here, our prime focus is on providing uniformly good monitoring performance over all untrusted pairs and hence MN $m$ employs   ZF combining for the UTs in $\Wm$ and MR combining for the UTs in $\Sm$.  In this case, the intra-group interference between UTs $\in
\Wm$ is actively cancelled, while the inter-group interference
between UTs $\in \Wm$ and UTs $\in \Sm$ is tolerated. 
We note that the number of  antennas at each MN  must meet
the requirement $\Ntx \geq |\Wm|+1$.
The local combining vector constructed by MN $m$ for UT $k \in \Wm$ is given by
 %=========
 \begin{align} \label{eq:ZF_prec}
     \vmku=\vmkz=\gamumk{\hGm\big[(\hGm)^\dag\hGm\big]^{-1}\qe_{k}},
 \end{align}
 %========
where $\hGm$ is an $N \times |\Wm|$ collective channel estimation matrix from all the UTs in $\Wm$ to MN $m$ as $\hGm=[\hat{\qg}_{mk}^\ul: k \in \Wm]$ and $\qe_{k}$ is the $k$-th column of $\qI_K$. 
Hence, for any  pair of UTs $k$ and  $\ell \in \Wm$ we have
%===
\begin{align}\label{eq:ZF_prec2}
(\vmkz)^\dag\hgmlu = \begin{cases} \displaystyle \gamumk & \text {if } k=\ell,\\ \displaystyle 0 & \text {otherwise}.
\end{cases}
\end{align}
%====
Moreover,   the MR combining vector constructed by MN $m$ for UT $k \in \Sm$ is given in~\eqref{eq:MR_PREC}.
Therefore,  by applying ZF combining for UTs $\in \Wm$ and MR combining
for UTs $\in \Sm$, \eqref{eq:rul} can be rewritten as
%=========================
\begin{align}\label{eq:rul_pzf}
&r_{k}^{\PZF}=\sum_{m \in M} \alphmk (\vmkpzf)^\dag\qy_{m}^{\ul} \nonumber
 \\
&= \text{DS}_k ^\PZF  s_k^\ul
+\sum\limits_{\ell \in {\mathcal{K}}\backslash k }{{\text{UI}}_{\ell k} ^\PZF}s_\ell^\ul
+ \sum\limits_{\ell \in {{\mathcal{K}}}} \text{MI}_{\ell k 
 }^{\PZF}s_\ell^\dl
+ {{\text{AN}}_k ^\PZF},
\end{align}
%=======================
where
%=========================
\begin{align}
\text{DS}_k ^\PZF=&  \sqrt{\rho _\UT}\Big(\sum\limits_{m \in \Zk} \alphmk(1-a_m)  {(\vmkz)^\dag}\gmku+\nonumber\\
&\quad\sum\limits_{m \in \Tk} \alphmk(1-a_m)  {(\vmkm)^\dag}\gmku \Big)  
\nonumber\\
{{\text{UI}}_{\ell k} ^\PZF}=&\sqrt{\rho _\UT} \Big(\sum\limits_{m \in \Zk}\alphmk (1-a_m) {(\vmkz)^\dag}\gmlu + \nonumber\\
&\quad\sum\limits_{m \in \Tk}\alphmk (1-a_m) {(\vmkm)^\dag}\gmlu  \Big)
\nonumber\\ 
{\text{MI}_{\ell k 
 }^{\PZF}} \! \!=& \!\sqrt{\rho _\dl}\!\! \sum\limits_{i \in {\mathcal{M}}}\!\!  \!a_i\sqrt{\theta_{i\ell}}  \Big(\!\!\!\sum\limits_{m \in \Zk}\!\!  \alphmk(1\!-\!a_m) 
 {(\vmkz)^\dag}{\qF_{mi}}{({\mathbf{\hat g}}_{i\ell}^\dl)^{\ast}}+\nonumber\\
 &\quad\sum\limits_{m \in \Tk}  \alphmk(1\!-\!a_m) 
 {(\vmkm)^\dag}{\qF_{mi}}{({\mathbf{\hat g}}_{i\ell}^\dl)^{\ast}}\!\Big)
\nonumber
\\
{{\text{AN}}_k ^\PZF}=&  \sum\limits_{m \in \Zk}\alphmk(1-a_m) {(\vmkz)^\dag  \qw_{m}^{\ul}}+\nonumber\\
&\quad\sum\limits_{m \in \Tk}\alphmk(1-a_m) {(\vmkm)^\dag  \qw_{m}^{\ul}},
\end{align}
%=======================
where $\Zk$ and $\Tk$ denote the set of indices of MNs that assign the $k$-th UT into $\Wm$ for  ZF combining and the set of indices of MNs that assign $k$-th UT into $\Sm$ for MR combining, respectively, as
$\Zk \triangleq \{m: k \in \Wm, m=1, \ldots, M\} $ and
$\Tk \triangleq \{m: k \in \Sm, m=1, \ldots, M\}$,
%=====
with $\Zk\cap \Tk =\emptyset$ and $\Zk\cup \Tk = \mathcal{M}$.

%%%%%%%%%%%%%%%%%%%%%%%%%%%%%%%%%%%%%%%%%%%%%%%%%%%%%%%%%%%%%%%%%%%%%
%==============================================
%%%%%%%%%%%%%%%%%%%%%%%%%%%%%%%%%%%%%%%%%%%%%%%%%%%%%%%%%%%%%%%%%%%%%
\section{Performance Analysis}
In this section, we derive the effective SINR of the untrusted communication links as well as the effective SINR for observing in conjunction with MR and PZF combining schemes. We also investigate the potential of using large number of MNs in either the observing  or jamming mode to cancel the inter-untrusted user interference and to enhance the energy efficiency, respectively.
%%%%%%%%%%%%%%%%%%%%%%%%%%%%%%%%%%%%%%%%%%%%%%%%%%%%%%%%%%%
\subsection{Effective SINR of the Untrusted Communication Links}
%%%%%%%%%%%%%%%%%%%%%%%%%%%%%%%%%%%%%%%%%%%%%%%%%%%%%%%%%%%
%%---------------
We define the effective noise as
%%-------------------
\begin{align}~\label{eq:ykdl3}
\tilde w_{k}^{\UR}
=&   \sqrt{\rho_\UT} \sum_{\ell\in \mathcal{K}, \ell \neq k}h_{\ell k} s_{\ell}^{\ul}
+\sqrt{\rho_\dl} \sum_{m \in \mathcal{M}}
\sum_{k'\in\mathcal{K}} a_m\times
\nonumber\\
 &
\quad\sqrt{\theta_{mk'}}
\left(\gmkd\right)^T\left(\hgmkpd\right)^*
s_{k'}^{\dl}\!+w_{k}^{\UR},
\end{align}
%---------------
and reformulate the  signal received at UR $k$ in~\eqref{eq:ykdl}  as
 %----------------------
 \begin{align}~\label{eq:ykdl2}
 y_k^{\UR}
 = \sqrt{\rho_\UT} h_{kk}s_{k}^{\ul}+\tilde w_{k}^{\UR}.
 \end{align}
%%------------
 Since $s_{\ell}^{\ul}$ is independent of $s_{k}^{\ul}$ for any $\ell\! \neq\! k$,  the first term of the effective noise in~\eqref{eq:ykdl3} is uncorrelated with the first term in~\eqref{eq:ykdl2}. Moreover, the second and third terms of~\eqref{eq:ykdl3} are uncorrelated with the first term of~\eqref{eq:ykdl2}. Therefore, the effective noise $\tilde w_{k}^{\UR}$ and the input RV $x_{k}^{\ul}$  are uncorrelated. Accordingly, we now obtain a closed-form expression for the effective  SINR of the untrusted link $k$.
%=====================
\begin{proposition}\label{Theorem:SE:SUS}
The effective SINR  of the untrusted link $k$ can be formulated as 
%================
\begin{align}\label{eq:SINRSR}
    \SINR_{\UR,k}  (\qa, \boldsymbol \theta) = \frac{ \rho_{\UT}|  h_{kk}|^2}{\xi_k(\qa,\boldsymbol{\theta})},
\end{align}
%----------------------
where
%-----------------------
\begin{align}\label{eq:zeta}
\xi_k(\qa,\boldsymbol{\theta})
=& \rho_{\UT}\!\!\!\sum_{\ell\in \mathcal{K}\setminus k} \!\betakldu\!+{\rho_\dl}\Ntx
\sum_{k'\in\mathcal{K}}\sum_{m \in \mathcal{M}}
 a_m\theta_{mk'}\betamkd\gamdmkp  
 \nonumber
 \\
&\hspace{3em} 
\nonumber\\
 &\quad
 +{\rho_\dl}\Ntx^2\Big(\sum_{m \in \mathcal{M}}a_m \sqrt{\theta_{mk}} \gamdmk\! \Big)^2+1,
 \end{align}
%-----------------------
with
$\qa \triangleq \{a_m\}$ and $\boldsymbol{\theta}\triangleq \{\theta_{mk}\}$, $\forall m,k$,  respectively.
\end{proposition}
%===================
\begin{proof}
See Appendix~\ref{ProofTheorem:SE:SUS}.
\end{proof}
%=====================
%=================================
\subsection{Effective SINR for Observing}
%==================================
The CPU detects the observed information $s_{k}^{\ul}$ from $r_{k}^{\ul}$ in~\eqref{eq:rul}. We assume that it does not have instantaneous CSI knowledge of the observing, jamming, and untrusted channels and  uses only statistical CSI when performs detection.  To calculate the effective SINR for the $k$-th untrsuted link, we use the  popular bounding technique, known as the hardening bound or the use-and-then-forget (UatF) bound~\cite{Hien:cellfree}\footnote{ This bound can be used for the scenarios, where the codeword spans over the time and frequency domains, i.e., across multiple coherence times and coherence bandwidths. This is practical   and it is widely supported in the literature of ergodic rate and capacity analysis~\cite{marzetta2016fundamentals,Emil:Book:2017}.}.
 In particular, we first rewrite the aggregated received signal for UT $k$ at the CPU as
%====
\begin{align}\label{eq:rulMR}
	r_{k}^{\CS}= &\mathbb{E}\{ {\text{DS}_k^\CS}\} s_k^\ul+\nonumber\\
 &~\text{BU}_k ^\CS s_k^\ul+\sum\limits_{\ell \in {\mathcal{K}}\backslash k }{{\text{UI}}_{\ell k} ^\CS}s_\ell^\ul+\sum\limits_{\ell \in {{\mathcal{K}}}}{\text{MI}_{\ell k} ^{\CS}}s_\ell^\dl+{{\text{AN}}_k ^\CS},
\end{align}
%----------------
where  $\text{BU}_k ^\CS ={\text{DS}_k^\CS}- \mathbb{E}\{{\text{DS}_k^\CS} \}$ reflects the beamforming gain uncertainty, while  the superscript “cs” refers to the “combining scheme”, $\CS = \{\MR, \PZF\}$. The CPU effectively encounters a deterministic channel ($\mathbb{E}\{ {\text{DS}_k^\CS}\}$) associated with some unknown noise. Since $s_k$ and $s_\ell$ are uncorrelated for any $\ell \neq k$, the first term in~\eqref{eq:rulMR} is uncorrelated with the third and forth terms. Additionally, since $s_k$ is independent of $\text{BU}_{k}$, the first and second terms are also uncorrelated. The fifth term, i.e., the noise, is independent of the first term in~\eqref{eq:rulMR}. Accordingly, the sum of the second, third,  fourth, and fifth terms in~\eqref{eq:rulMR} can be collectively considered as an uncorrelated effective noise. Therefore, the received SINR of observing the untrusted link $k$ can be formulated as
% %==========
\begin{align}\label{eq:SINRST}
&\SINR_{\ul,k}^\CS = \nonumber\\
&\frac{|\mathbb{E}\{{{\text{DS}}_k ^\CS}\}|^2}{{\mathbb{E}\big\{\vert {\text{BU}}_k ^\CS\vert^2\big\}
\!+\!\!\!\!\!\sum\limits_{{\ell \in {\mathcal{K}}\backslash k }} \!\!\!\!\mathbb{E}\big\{\vert {\text{UI}}_{\ell k} ^\CS\vert^2\big\} 
\!+\! \!\!\sum\limits_{{\ell \in {\mathcal{K}} }} \!\!\mathbb{E}\big\{\vert 
 {\text{MI}}_{\ell k} ^\CS\vert^2\big\} 
 \!+\! \mathbb{E}\big\{\vert {\text{AN}}_k ^\CS}\vert^2\big\}}.
\end{align}
By calculating the
corresponding expected values in~\eqref{eq:SINRST},  the SINR observed for the untrusted link $k$ for MR and PZF combining schemes can be obtained as in the following propositions.
%=======================
%----------------------
\begin{proposition}\label{Theorem:SE:CPU}
The received SINR for the $k$-th untrusted link at the CPU for MR combining is given by
%-----------------
\vspace{-0.1em}
\begin{align}\label{eq:SINRMA}
&\SINR_{\ul,k}^{\MR}(\ALPHA,\qa,\boldsymbol{\theta})=\frac{
	\Nrx \rho_{\UT} \Big(\sum\limits_{\substack{m\in\mathcal{M}}} \alphmk (1-a_m )  \gamumk \Big)^2
	}
	{		 \mu_k^\MR\!+\!
		\rho_\dl\Ntx\!\!\!
		\sum\limits_{\substack{i\in\mathcal{M}}}\!
		\sum\limits_{\ell\in\mathcal{K}}\!
		\!
		a_{i}\theta_{i\ell}  \gamma_{i\ell}^{\dl} \varrho_{i  k}^\MR
},
\end{align}
with
%--------------
\vspace{-0.4em}
\begin{align*}
\mu_k^\MR\triangleq& \!\!\!\sum\limits_{\substack{m\in\mathcal{M}}}\! \!
		\!\big(\rho_{\UT}\sum\limits_{\ell\in\mathcal{K}} \alphmk^2
				(1\!\!-\!a_m) 
	\beta_{m\ell}^{\ul}
		\gamma_{mk}^{\ul}
	+\alphmk^2
		(1\!-\!a_m) \gamumk\big),\nonumber\\
&\varrho_{i  k}^\MR\triangleq 		\!\sum\limits_{\substack{m\in\mathcal{M}}}\!\alphmk^2	(1-a_m)  \gamumk \beta_{mi},
  \end{align*}
  %-----------------
  where $\ALPHA=\{\alpha_{mk}\},~\forall m ,k.$
%====================

\end{proposition}
\begin{proof}
See Appendix~\ref{ProofTheorem:SE:CPU}.
\end{proof}
%----------------------
\begin{proposition}\label{Theorem:SEZF:CPU}
The received SINR for the $k$-th untrusted link at the CPU for PZF combining is given by
%-----------------
%-----------------
\begin{align}\label{eq:SINRMA_ZF1}
&\SINR_{\ul,k}^{\PZF}(\ALPHA,\qa,\boldsymbol{\theta})=\nonumber
\\
&\frac{
	\rho_{\UT} \Big(\sum\limits_{\substack{m\in\Zk}} \alphmk(1-a_m )  {\gamumk} +\Ntx\!\!\sum\limits_{\substack{m\in\Tk}} \!\alphmk(1-a_m )  {\gamumk}\!\Big)^2
	}
	{\mu_k^\PZF+
		\rho_\dl\Ntx
		\sum\limits_{\substack{i\in\mathcal{M}}}
		\sum\limits_{\ell\in\mathcal{K}}
		 a_{i}\theta_{i\ell}  {\gamma_{i\ell}^{\dl}}\varrho_{i k}^\PZF
},
\end{align}
with
%--------------
\begin{align*}
&\mu_k^\PZF\triangleq\rho_{\UT}
		\! \sum\limits_{\substack{m\in\Zk}}
		 \sum\limits_{\ell\in\mathcal{K}}\!
				\alphmk^2(1-a_m) 
	\frac{\gamumk(\beta_{m\ell}^{\ul}-
		\gamuml)}{\Ntx-|\Wm|}\nonumber\\
 &+\rho_{\UT}\Ntx
		\! \sum\limits_{\substack{m\in\Tk}}
		 \sum\limits_{\ell\in\mathcal{K}}\!
				\alphmk^2(1-a_m) 
	{\gamumk\beta_{m\ell}^{\ul}}\nonumber
  \\
  &	+\sum\limits_{\substack{m\in\Zk}}
		\alphmk^2(1\!\!-\!a_m)\frac{\gamumk}{\Ntx-|\Wm|}+\Ntx\!\!\sum\limits_{\substack{m\in\Tk}}\!\!
		\alphmk^2(1\!\!-\!a_m)\gamumk.
  \end{align*}
%-----------------
%-------------------
\begin{align}
 \varrho_{i  k}^\PZF\triangleq& \sum\limits_{\substack{m\in\Zk}}
		\alphmk^2(1-a_m) \frac{\gamumk \beta_{mi}}{\Ntx-|\Wm|} \nonumber
  \\
  &\quad+\Ntx\!\sum\limits_{\substack{m\in\Tk}}
		\alphmk^2(1-a_m)  {\gamumk \beta_{mi}}.\nonumber
\end{align}

\end{proposition}
\begin{proof}
See Appendix~\ref{ProofTheorem:SEZF:CPU}.
\end{proof}
%=================================
%=================================================================
\subsection{Large-$M$ Analysis}
In this subsection, we provide some insights into the performance of CF-mMIMO surveillance systems when the number $M_\ul$ of MNs in observing mode  or the number $M_\dl$ of MNs in jamming mode  is very large.  The asymptotic results are presented for MR combining, while the same method and insights can be obtained for PZF combining. 

%=================================================================================
%==========================================================================
\emph{1) Using Large Number of MNs in Observing Mode, $ M_\ul \rightarrow \infty$}:\label{subsec:MOinfy}
%===========
Assume that the number of untrusted pairs, $K$, is fixed. For any finite $M_\dl$, as $M_\ul\rightarrow \infty$, we have the following results for the  signal received at the CPU for observing UT $k$ employing the MR combining. 
%==============================================================
By using Tchebyshev’s theorem~\cite{cramer70}, we obtain  
%=====
\begin{align}
& \frac{1}{M_\ul}{\text{DS}}_k ^\MR s_k ^\ul \!-\! \frac{1}{M_\ul}\sqrt {{\rho _\UT}}\!\! \sum\limits_{m \in {\mathcal{M}}}\!\!N\alphmk(1\!-\!a\!_m)\gamumk 
s_k ^\ul \! \xrightarrow[{M_\ul \to \infty }]{P}0,
\nonumber\\
& \frac{1}{M_\ul} \sum\limits_{\ell \in {\mathcal{K}}\backslash k } \text{UI}_{\ell k}^\MR s_\ell^\ul\xrightarrow[{M_\ul \to \infty }]{P}0, \nonumber\\
 &\frac{1}{M_\ul}{\text{AN}}_k ^\MR\xrightarrow[{M_\ul \to \infty }]{P}0,\nonumber\\
 &\frac{1}{M_\ul} \sum\limits_{\ell \in {\mathcal{K}}} \text{MI}_{\ell k} ^\MR s_\ell^\dl\xrightarrow[{M_\ul \to \infty }]{P}0,
\end{align}
%=============
where $\xrightarrow[{M_\ul \to \infty }]{P}0$  shows convergence in probability when $M_\ul\rightarrow \infty$. 
%==========
The above expressions show that when $M_\ul \to \infty$,
the observed signal includes only the desired signal. The monitoring performance can improve without limit by using more MNs in observing mode.
%=================================================================================
%==========================================================================

\emph{2) Using Large Number of MNs in Jamming Mode, $ M_\dl \rightarrow \infty$}:
%===========
Assume that  the number of MNs in jamming mode goes to infinity, while transmit power of each MN in jamming mode is scaled with $ M_\dl$ according
to $\rho_\dl= \frac{E_\dl}{M_\dl}$, where $E_\dl$  is fixed.
The aggregated received signal expression in~\eqref{eq:rulMR}  for the MR combining scheme  shows that   $\text{MI}_{\ell k} ^\MR$ is
dependent on $M_\dl$; however $\text{DS}_k ^\MR,  \text{UI}_{\ell k} ^\MR$,  and $\text{AN}_k ^\MR$  are constant with
respect to $M_\dl$.  Now, let us assume that the number of untrusted pairs, $K$, is fixed. For any finite $M_\ul$, when $M_\dl\rightarrow \infty$ and $\rho_\dl= \frac{E_\dl}{M_\dl}$, we have 
%==========
\vspace{-0.5em}
\begin{align}
&\text{MI}_{\ell k} ^{\MR}= \sqrt{\frac{E _\dl}{M_\dl}}\sum\limits_{m \in {\mathcal{M}}}\!\!\alphmk(1-a_m)  
 {(\hgmku)^\dag} \qt_{\ell m} ,
\end{align}
%----------------
where
%-----------
\vspace{-0.5em}
\begin{align}
\qt_{\ell m} &= \sum\limits_{i \in {\mathcal{M}}}a_i\sqrt{\theta_{i\ell}}{\qF_{mi}}{({\mathbf{\hat g}}_{i\ell}^\dl)^{\ast}}\nonumber\\
&= \frac{\sum\limits_{i \in {\mathcal{M}}}a_i\sqrt{\theta_{i\ell}}{\qF_{mi}}{({\mathbf{\hat g}}_{i\ell}^\dl)^{\ast}}}{\sqrt{\sum\limits_{i \in {\mathcal{M}}}a_i\theta_{i\ell}{\beta_{mi}}
{\Vert{\mathbf{\hat g}}_{i\ell}^\dl\Vert^2}}}
{\sqrt{\sum\limits_{i \in {\mathcal{M}}}a_i\theta_{i\ell}{\beta_{mi}}
{\Vert{\mathbf{\hat g}}_{i\ell}^\dl\Vert^2}}}.
\end{align}
%----------
Now, let use define 
%----------------
\vspace{-0.5em}
\begin{align*}
  \qz_{\ell m} \triangleq \frac{\sum\limits_{i \in {\mathcal{M}}}a_i\sqrt{\theta_{i\ell}}{\qF_{mi}}{({\mathbf{\hat g}}_{i\ell}^\dl)^{\ast}}}{\sqrt{\sum\limits_{i \in {\mathcal{M}}}a_i\theta_{i\ell}{\beta_{mi}}
{\Vert{\mathbf{\hat g}}_{i\ell}^\dl\Vert^2}}}. 
\end{align*}
%--------------------
For given $\{\mathbf{\hat g}_{i\ell}^\dl$\}, $\qz_{\ell m}$  is distributed on  $\mathcal{CN} (0, \qI_\Ntx)$. Therefore, $\qz_{\ell m}\sim \mathcal{CN} (0, \qI_\Ntx)$ is independent of  $\{\mathbf{\hat g}_{i\ell}^\dl$\}. Thus, we have
%--------------
\vspace{-0.5em}
\begin{align}\label{eq:qtell}
\qt_{\ell m} &={\sqrt{\sum\limits_{i \in {\mathcal{M}}}a_i\theta_{i\ell}{\beta_{mi}}
{\Vert{\mathbf{\hat g}}_{i\ell}^\dl\Vert^2}}}\qz_{\ell m}.
\end{align}
%----------------
 %====
By using \eqref{eq:qtell} and Tchebyshev’s theorem, we obtain 
 %-----
\begin{align}
& \frac{1}{\sqrt{M_\dl}}\qt_{\ell m} - \sqrt 
{\frac{1}{M_\dl}{\sum\limits_{i \in {\mathcal{M}}}  }\Ntx a_i\theta_{i\ell}\beta_{mi}\gamma_{i\ell}^\dl}\qz_{\ell m} \xrightarrow[{M_\dl \to \infty }]{P}0.
\end{align}
 %-----------
 As a result,
\begin{align}\label{eq:MI_asmp}
&\sum\limits_{\ell \in \mathcal{K}}{\text{MI}}_{\ell k} ^\MR s_\ell^\dl -\sqrt {{E_\dl}} \sum\limits_{\ell \in {{\mathcal{K}}}} \sum\limits_{m \in {\mathcal{M}}}\alphmk (1-a_m)(\hgmku)^\dag\times\nonumber
\\
&\hspace{5em}\sqrt{\frac{1}{M_\dl}{\sum\limits_{i \in {\mathcal{M}}}  }\Ntx a_i\theta_{i\ell}\beta_{mi}\gamma_{i\ell}^\dl} \qz_{\ell m}s_\ell^\dl \xrightarrow[{M_\dl \to \infty }]{P}0.
\end{align}
%----------------------
%==========
Expression in~\eqref{eq:MI_asmp} shows that for large $M_\dl$, we can reduce the transmitted jamming power of each MN in jamming mode proportionally to $1/M_\dl$,  while maintaining the  given SINR for observing. At the same time, from \eqref{eq:ykdl}, by using again the Tchebyshev’s theorem, we can show that the SINR for the untrusted commnication links goes to $0$, as $M_\dl$ goes to infinity.  This verifies 
the potential of using a large number of  MNs in jamming mode to save power and, hence, enhance the energy efficiency of CF-mMIMO surveillance systems.
%**********************************************************
%=========================================================
\subsection{Monitoring Success Probability}
%--------------------------------------------------
To achieve  reliable detection at UR $k$, UT $k$ varies its transmission rate according to the prevalent $\SINR_{\UR,k}$. In particular the $k$-th UR provides SINR feedback to the $k$-th UT concerning its perceived channel quality. Based on this feedback, the UT dynamically adapts its modulation and coding scheme. Higher SINR values allow for higher data rates, while lower SINR values  necessitate lower data rates to maintain reliable communication.
Hence, if $\SINR_{\ul,k}^\CS\geq\SINR_{\UR,k}$,  the CPU can also reliably detect the information of the untrusted link $k$. On the other hand, if $\SINR_{\ul,k}^\CS\leq\SINR_{\UR,k}$, the CPU may detect this information at a high probability of error. Therefore, the following indicator function can be designed for characterizing the event of
successful monitoring at the CPU~\cite{Zhong:TWC:2017}
%---------------
\vspace{-0.5em}
\begin{align*}
X^\CS_k = \begin{cases} \displaystyle 1 & \text {if } \SINR_{\ul,k}^\CS\geq\SINR_{\UR,k},\\ \displaystyle 0 & \text {otherwise},
\end{cases}
\end{align*}
%------------------
where $X^\CS_k=1$ and $X^\CS_k=0$ indicate the monitoring success and failure events for the untrusted link $k$,  respectively. 
Thus, a suitable performance metric for  monitoring
each untrusted communication link $k$ is the MSP, $\mathbb {E}\{X^\CS_k\}$, defined as 
%--------------------
\begin{equation} \label{eq:non-outage}
\mathbb {E}\{X^\CS_k\}= \text {Pr}\left (\SINR^\CS_{\ul,k}\geq\SINR_{\UR,k}\right)\!.
\end{equation}
%--------------------
From~\eqref{eq:SINRSR},~\eqref{eq:SINRMA}, and~\eqref{eq:non-outage} we have
%---------------
%\vspace{-0.2em}
\begin{equation}
\mathbb {E}\{X^\CS_k\}=\mathbb {P}\Big(| h_{kk}|^2\leq \frac{\SINR^\CS_{\ul,k} \xi_k}{\rho_{\UT}}\Big).
\end{equation}
%---------------
Using the cumulative distribution function (CDF) of the exponentially distributed RV  $|h_{kk}|^2$, the MSP  of our CF-mMIMO surveillance system can be expressed in closed form as
%------------
%\vspace{-0.2em}
\begin{equation} \label{eq:Nonout}
\mathbb {E}\{X^\CS_k\}=1-\exp \left(-\frac{\SINR^\CS_{\ul,k} \xi_k}{\betakkdu\rho_{\UT}}\right). \end{equation}

%================================================================
%==============================================================
\section{Max-Min MSP Optimization}
\label{sec:SE}
%----------------------------------------
In this section, we aim for maximizing the lowest probability of successful monitoring by  optimizing the MN-weighting  coefficients $\ALPHA$, the observing and jamming mode assignment vector $\aaa$, and the power control coefficient vector $\THeta$ under the constraint of the transmit power at each MN in~\eqref{DL:power:cons}. 
More precisely,  we formulate an optimization problem as
%-----------------------------------------------------------------
\vspace{-0.5em}
\begin{subequations}\label{P:SE21}
	\begin{align}
		\underset{\{ \ALPHA,\qa,\boldsymbol \theta\}}{\mathrm{max}}\,\, &\hspace{1em}
		\underset{k\in\mathcal{K}} \min \,\,\mathbb {E}\{X^\CS_k ({\boldsymbol{\ALPHA,\qa,\theta}}) \}
		\\
		\mathrm{s.t.} \,\,
		& \hspace{1em} a_m\sum_{k\in\mathcal{K}} \gamdmk \theta_{mk} \leq \frac{1}{\Ntx},~~  m\in\mathcal{M},\label{opt:cons1}\\
  &\hspace{1em}\theta_{mk}\geq 0, ~~ m\in\mathcal{M},~k\in\mathcal{K}, \label{opt:cons2}\\
    &\hspace{1em}0\leq\alphmk\leq 1, ~~ \quad \forall k, m,\label{opt:cons3}\\
  &\hspace{1em}a_{m} \in \{0,1\}, ~~ m\in\mathcal{M}. 
 \label{opt:cons4}
   \end{align}
\end{subequations}
%--------------------------------------------------------------------
By substituting~\eqref{eq:Nonout} into~\eqref{P:SE21}, the optimization problem \eqref{P:SE21} becomes
%===============
\vspace{-0.5em}
\begin{subequations}\label{P:SE2}
	\begin{align}
		\underset{\{ \ALPHA,\qa,\boldsymbol \theta\}}{\mathrm{max}}\,\, &\hspace{1em}
		\underset{k\in\mathcal{K}} \min \,\,1-\exp \left(-\frac{\SINR^\CS_{\ul,k}({\boldsymbol{\ALPHA,\qa,\theta}}) \xi_k({\boldsymbol{\qa,\theta}})}{\betakkdu\rho_{\UT}}\right) 
		\\
		\mathrm{s.t.} \,\,
		& \hspace{1em} a_m\sum_{k\in\mathcal{K}} \gamdmk \theta_{mk} \leq \frac{1}{\Ntx},~~  m\in\mathcal{M},\label{opt:cons1}\\
  &\hspace{1em}\theta_{mk}\geq 0, ~~ m\in\mathcal{M},~k\in\mathcal{K}, \label{opt:cons2}\\
    &\hspace{1em}0\leq\alphmk\leq 1, ~~ \quad \forall k, m,\label{opt:cons3}\\
  &\hspace{1em}a_{m} \in \{0,1\}, ~~ m\in\mathcal{M}. 
 \label{opt:cons4}
   \end{align}
\end{subequations}
%================
By using the fact that  $1-\exp(-x)$ is a monotonically increasing function of $x$ and since $\betakkdu$ and $\rho_{\UT}$ are fixed values independent of the optimization variables, the problem~\eqref{P:SE2} is equivalent to the following problem
%--------------------------------------------------------------------
%------------
\vspace{-0.5em}
\begin{subequations}\label{P:SE3}
\	\begin{align}
		&\underset{\{\ALPHA,\qa,\boldsymbol \theta\}}{\mathrm{max}}\,\, \hspace{1em}
		\underset{k\in\mathcal{K}} \min \,\,\SINR^\CS_{\ul,k}(\ALPHA,\qa,\boldsymbol \theta)\xi_k(\qa,\boldsymbol \theta)
		\\
		% \mathrm{s.t.} \,\,		
		% &\hspace{1em} a_m\sum_{k\in\mathcal{K}} \gamdmk \theta_{mk} \leq \frac{1}{\Ntx}, ~~  m\in\mathcal{M}, \\
  % &\hspace{1em}\theta_{mk}\geq 0 ~~m\in\mathcal{M},~k\in\mathcal{K}.
  &~\text {s.t.} \hspace{2em}~\eqref{opt:cons1}-\eqref{opt:cons4}.
	\end{align}
 \end{subequations}
 %
%-----------------
Problem~\eqref{P:SE3}   has as a tight coupling of the MN-weighting  coefficients $\ALPHA$, of  the observing and jamming mode assignment vector $\aaa$, and of the power control coefficient vector $\THeta$. In particular, observe from~\eqref{eq:SINRMA} and~\eqref{eq:SINRMA_ZF1}, that
in $\SINR^\CS_{\ul,k}$, the power coefficients $\theta_{i\ell}$ are coupled with the    mode assignment parameters $a_{i}$. Furthermore, the mode assignment parameters $a_m$ are also coupled with the  MN-weighting coefficients  $\alphmk$. Therefore, problem~\eqref{P:SE3}  is not jointly convex in terms of $\ALPHA$, of the power allocation coefficients $\boldsymbol \theta$, and of the mode assignment  $\qa$.  This issue makes the max-min MSP problem technically challenging,  hence it is difficult to find its optimal solution.
Therefore, instead of finding the optimal solution, we aim for finding a suboptimal solution. To  this end,  we conceive  a heuristic greedy method for MN mode assignment, which simplifies the computation, while providing a significant successful monitoring performance gain. In addition, for a given mode assignment, the max-min MSP problem can be formulated as the following optimization framework:
%-------------
\vspace{-0.5em}
\begin{subequations}\label{P:SE4}
\	\begin{align}
		&\underset{\{\ALPHA,\boldsymbol \theta\}}{\mathrm{max}}\,\, \hspace{1em}
		\underset{k\in\mathcal{K}} \min \,\,\SINR^\CS_{\ul,k}(\ALPHA,\boldsymbol \theta)\xi_k(\boldsymbol \theta)
		\\
		% \mathrm{s.t.} \,\,		
		% &\hspace{1em} a_m\sum_{k\in\mathcal{K}} \gamdmk \theta_{mk} \leq \frac{1}{\Ntx}, ~~  m\in\mathcal{M}, \\
  % &\hspace{1em}\theta_{mk}\geq 0 ~~m\in\mathcal{M},~k\in\mathcal{K}.
  &~\text {s.t.} \hspace{2em}~\eqref{opt:cons1}-\eqref{opt:cons3}.
	\end{align}
 \end{subequations}
%-----------------
Problem~\eqref{P:SE4} is not jointly convex in terms of $\ALPHA$ and power allocation $\boldsymbol \theta$. To tackle this non-convexity issue, we cast the  optimization problem~\eqref{P:SE4} into two sub-problems: the MN-weighting  control problem  and the power allocation problem. To obtain a solution for problem~\eqref{P:SE4}, these sub-problems are alternately solved, as outlined in the following subsections.

%%%%%%%%%%%%%%%%%%%%%%%%%%%%%%%%%%%%%%%%%
\begin{algorithm}[!t]
\caption{Greedy MN Mode Assignment}
\begin{algorithmic}[1]
\label{alg:Grreedy} 
\STATE
\textbf{Initialize}: Set  $\mathcal{A}_{\ul}=\mathcal{M}$ and $\mathcal{A}_{\dl}=\emptyset$. Set iteration index $i=0$.
\STATE Calculate $\Pi^\star[i]=  		\underset{k\in\mathcal{K}} \min \,\,\mathbb {E}\{X^\CS_k (\mathcal{A}_{\ul}, \mathcal{A}_{\dl})\}$
\REPEAT
\FORALL{$m \in \mathcal{A}_{\ul}$}
\STATE Set $\mathcal{A}_{s}=\mathcal{A}_{\ul} \setminus m$.
\STATE  Calculate $\Pi_m=  		\underset{k\in\mathcal{K}} \min \,\,\mathbb {E}\{X^\CS_k (\mathcal{A}_{s}, \mathcal{A}_{\dl}\bigcup m)\}$\\
\ENDFOR
\STATE Set $\Pi^\star[i+1]= \underset{m\in\mathcal{A}_{\ul}} \max \,\,\Pi_m$\\
\STATE $e=|\Pi^\star[i+1]- \Pi^\star[i]|$ 
\IF{$e \geq e_{\min}$ }
% \STATE{Update  $\mathcal{A}_{\ul}=\{\mathcal{A}_{\ul}\bigcup m^{\star}\}$ }
% \ELSE
\STATE Select MN $m^\star=\argmax_{m\in\mathcal{A}_{\ul}}\{\Pi_m\}$
\STATE {Update $\mathcal{A}_{\dl}=\{\mathcal{A}_{\dl}\bigcup m^{\star}\}$ and $\mathcal{A}_{\ul}=\mathcal{A}_{\ul}\setminus m^{\star}$}
\ENDIF
\UNTIL{ $e < e_{\min}$ }
\RETURN $\mathcal{A}_{\ul}$ and $\mathcal{A}_{\dl}$, i.e., the indices of MNs in observing mode and jamming mode, respectively.
\end{algorithmic}
\end{algorithm}
\setlength{\textfloatsep}{0.2cm}
%---------------------------------------
%================================
%\vspace{-0.9cm}
\subsection{Greedy MN Mode Assignment for  Fixed Power Control and MN-Weighting  Control }
Let $\mathcal{A}_{\ul}$ and $\mathcal{A}_{\dl}$ denote the sets containing the indices of MNs in observing mode, i.e., MNs with $a_m=0$, and MNs in jamming mode, i.e., MNs with $a_m=1$, respectively. 
In addition, $\mathbb {E}\{X^\CS_k(\mathcal{A}_{\ul}, \mathcal{A}_{\dl})\}$ presents the dependence of the MSP  on the different choices of MN mode assignments. Our greedy algorithm of MN mode assignment is shown in \textbf{Algorithm~\ref{alg:Grreedy}}. All MNs are
initially assigned to observing mode, i.e., $\mathcal{A}_{\ul}=\mathcal{M}$ and $\mathcal{A}_{\dl}=\emptyset$. Then, in each iteration,  one MN switches into jamming mode for maximizing the minimum MSP~\eqref{eq:Nonout} among the untrusted links, until there is no more improvement.
%------------------------------------------------------------------
\vspace{-1em}
\subsection {Power Control for  Fixed  MN Mode Assignment and MN-Weighting  Control }
%-----------------------------------------
For the given MN mode assignment and MN-weighting coefficient control, the optimization problem~\eqref{P:SE3}  reduces to the power control problem. Using~\eqref{eq:zeta},~\eqref{eq:SINRMA},~\eqref{eq:SINRMA_ZF1} and~\eqref{P:SE4}, the max-min MSP problem is now formulated as
%------------------------
\vspace{-0.2em}
\begin{subequations}\label{eq:PAoptimization1}
\begin{align}
&\max _{\boldsymbol{\theta}}~\min \limits _{\forall k\in\mathcal{K}}  \frac {\xi_k(\boldsymbol{\theta})}
{	  \mu_{k}^\CS\!\!+\!\rho_\dl\Ntx\!
		\!\sum\limits_{\substack{i\in\mathcal{M}}}
		\sum\limits_{\ell\in\mathcal{K}}
		\!
		 a_{i}\theta_{i\ell}  \gamma_{i\ell}^{\dl} \varrho_{i  k}^\CS
		}. \\[0.5pt]
% &\text {s.t.}\quad a_m\sum_{k\in\mathcal{K}} \gamdmk \theta_{mk} \leq \frac{1}{\Ntx}, ~ \forall m\in\mathcal{M}, \\[0.5pt]
% &\qquad \theta_{mk} \geq 0, \quad \forall k\in\mathcal{K}, ~ \forall m\in\mathcal{M}. 
&\text {s.t.} \hspace{2em}~\eqref{opt:cons1},~\eqref{opt:cons2}.
\end{align}
\end{subequations}
%------------------
By introducing the slack variable $\zeta$, we  reformulate~\eqref{eq:PAoptimization1} as
%---------------
\vspace{-0.4em}
\begin{subequations}~\label{eq:PAoptimization}
\begin{align} 
&\max_{\{\boldsymbol{\theta}, \zeta\}}~ \zeta \\
&\text {s.t} \hspace{2em}	\rho_\dl\Ntx
		\sum\limits_{\substack{i\in\mathcal{M}}}
		\sum\limits_{\ell\in\mathcal{K}}
		 a_{i}\theta_{i\ell}   \gamma_{i\ell}^{\dl} \varrho_{i  k}^\CS
          - \frac{1}{\zeta}{\xi}_k(\boldsymbol{\theta}) \nonumber\\
          & \hspace{4em}+\mu_{k}^\CS\leq 0,
~~\forall k\in\mathcal{K}.
\label{eq:cons_zet}
 \\
% &\quad \quad \quad a_m\sum_{k\in\mathcal{K}} \gamdmk \theta_{mk}\leq \frac{1}{\Ntx}, ~ \forall m\in\mathcal{M}, \\
% &\quad \quad \quad  \theta_{mk} \geq 0, ~ \forall k\in\mathcal{K}, ~ \forall m\in\mathcal{M}. 
&\hspace{2em}~\eqref{opt:cons1},~\eqref{opt:cons2}.
\end{align}
\end{subequations}
%--------------
To arrive at a  computationally more efficient formulation, we use the inequality $\left(\sum_{m \in \mathcal{M}}\sqrt{\theta_{mk}}\gamdmk\right)^2 \geq \sum_{m \in \mathcal{M}}\theta_{mk}(\gamdmk)^2$ and replace the constraint~\eqref{eq:cons_zet}  by
%-----
\vspace{-0.4em}
\begin{align}
&\rho_\dl\Ntx
		\sum\limits_{\substack{i\in\mathcal{M}}}
		\sum\limits_{\ell\in\mathcal{K}}
		 a_{i}\theta_{i\ell}   \gamma_{i\ell}^{\dl}\varrho_{i k}^\CS
         - \frac{1}{\zeta}\tilde{\xi}_k(\boldsymbol{\theta}) +\mu_{k}^\CS  \leq 0,
\end{align}
%---------
where
%-----
\begin{align}\label{eq:zetat}
&\tilde{\xi}_k\!(\boldsymbol{\theta}) 
\!=\!\rho_{\UT}\!\!\!\sum_{\ell\in \mathcal{K}\setminus k} \!\betakldu\!+\!
{\rho_\dl}\Ntx
\!\sum_{k'\in\mathcal{K}}\sum_{m \in \mathcal{M}}\!
a_m\theta_{mk'}\betamkd\gamdmkp
\nonumber\\
&\qquad\quad
+{\rho_\dl} \Ntx^2\sum_{m \in \mathcal{M}}a_m\theta_{mk}(\gamdmk)^2
+
1.
 \end{align}
%-------
%----
Now, for a fixed $\zeta$, all the inequalities appearing in~\eqref{eq:PAoptimization} are linear,  hence the program~\eqref{eq:PAoptimization} is quasi–linear. 
Since the second constraint in~\eqref{eq:PAoptimization} is an
increasing functions of $\zeta$, the solution to the optimization problem is obtained by harnessing a line-search over
$\varrho_{ik}^\CS$ to find the maximal feasible value.
As a consequence, we use the  bisection
method in \textbf{Algorithm 2} to obtain the solution.

%---=================================
%------------------------------------------
%==================================
\begin{algorithm}[t]
\caption{Bisection Method for Max-Min Power Control }
\label{Alg:PA}
\begin{algorithmic}[1]
\STATE Initialization of $\zeta_{\min}$ and $\zeta_{\max}$, where $\zeta_{\min}$ and $\zeta_{\max}$ define a range of relevant values of the objective function in~\eqref{eq:PAoptimization1}. Initial line-search accuracy  $\epsilon$.
\REPEAT
    \STATE Set $\zeta:=\frac{\zeta_{\min}+\zeta_{\max}}{2}$. Solve the following convex feasibility program
    %----------------------
    \begin{align}\label{eq:Feasibility}
    \begin{cases} \hspace {0.0 cm}
&\!\!\!\!\!{{\mu_{k}^\CS\!+\!\rho_\dl\Ntx\!\!
		\sum\limits_{\substack{i\in\mathcal{M}}}
		\sum\limits_{\ell\in\mathcal{K}}
		\!
		 }} 
 a_{i}\theta_{i\ell}  \gamma_{i\ell}^{\dl}\varrho_{i k}^\CS\!-\!{\frac{1}{\zeta}\tilde{\xi}_k(\boldsymbol{\theta})\! \leq\! 0,
~\forall k\in\mathcal{K}},
  \\
&\!\!\!\!a_m\sum_{k\in\mathcal{K}} \gamdmk \theta_{mk}\leq \frac{1}{\Ntx}, ~ \forall m\in\mathcal{M}, \\
&\!\!\!\!\theta_{mk} \geq 0, ~ \forall k\in\mathcal{K}, ~ \forall m\in\mathcal{M}.
 \end{cases} 
\end{align}
%-
\STATE If problem~\eqref{eq:Feasibility} is feasible, then set $\zeta_{\min}:=\zeta$, else set $\zeta_{\max} :=\zeta$.
\UNTIL{ $\zeta_{\max}-\zeta_{\min}<\epsilon$ }
\end{algorithmic}
\end{algorithm}
%%=======================================
%%=================
\vspace{-0.4em}
\subsection {MN-Weighting  Control  for  Fixed    MN Mode Assignment and  Power Control}
The received SINR at the URs is independent of the MN-weighting  coefficients $\ALPHA$. Therefore, the  coefficients $\ALPHA$ can be obtained by independently maximizing the received  SINR of each untrusted link $k$ at the CPU. Therefore, the optimal  MN-weighting   coefficients for all UTs for the given transmit
power allocations and mode assignment, can be found by solving the following problem:
%------------
\begin{subequations}\label{P:LSFD}
\	\begin{align}
		&\underset{\ALPHA}{\mathrm{max}}\,\, \hspace{1em}
		 \,\,\SINR^\CS_{\ul,k}(\ALPHA)
		\\
  &~\text {s.t.}   \hspace{1em}0\leq\alphmk\leq 1, ~~ \quad \forall k, m. 
	\end{align}
 \end{subequations}
Let us introduce a pair of binary variables to indicate the group assignment for each UT $k$ and MN $m$ in our PZF combining scheme as
%==========
\begin{align*}
\dm^{\SZ} = \begin{cases} \displaystyle 1 & \text {if } m \in \Zk,\\ \displaystyle 0 & \text {otherwise},
\end{cases}
\qquad
\dm^{\OT} = \begin{cases} \displaystyle 1 & \text {if } m \in \Tk,\\ \displaystyle 0 & \text {otherwise}.
\end{cases}
\end{align*}
%==========
%-----------------
Then, to solve~\eqref{P:LSFD}, we use the following proposition:
%=======================
\begin{proposition} \label{theorem:lsfd}
 The optimal MN-weighting  coefficient vector, maximizing the   SINR  observed for the $k$-th untrusted link  can  be obtained as 
%======
\begin{align}
\ALPHA_k^\star=\diag(\qB_{1k}^\CS,\ldots,\qB_{Mk}^\CS)^{-1}\qc_k^\CS,
\end{align}
%====
where    $ \qB_{mk}^\CS=u_{mk}^\CS+ \rho_\dl\Ntx
		\sum\limits_{\substack{i\in\mathcal{M}}}
		\sum\limits_{\ell\in\mathcal{K}}
		 a_{i}\theta_{i\ell}  {\gamma_{i\ell}^{\dl}}v_{m i k}^\CS,$
%=====
$\qc_k^\CS=[c_{1k}^\CS,\cdots,c_{Mk}^\CS]$ with elements $c_{mk}^\MR = (1-a_m)\sqrt{\Ntx}\gamma_{mk}$ and 
%====
\begin{align*}
c_{mk}^\PZF = \begin{cases} \displaystyle (1-a_m)\gamma_{mk} & \text {if } m \in \Zk,\\ \displaystyle \displaystyle (1-a_m)\Ntx\gamma_{mk} & \text {if } {m \in \Tk},
\end{cases}
\end{align*}
%====

%--------------
\begin{align*}
&u_{m k}^\MR= \rho_{\UT} 
		 \sum\limits_{\ell\in\mathcal{K}} 
				(1-a_m) 
	\beta_{m\ell}^{\ul}
		\gamma_{mk}^{\ul}
	+
		(1-a_m) \gamumk,\nonumber\\
&v_{m i  k}^\MR= (1-a_m)  \gamumk \beta_{mi},
\\
&u_{mk}^\PZF=(1-a_m) \big(\rho_{\UT}
		 \sum\limits_{\ell\in\mathcal{K}}\dm^{\SZ}
					\frac{\gamumk(\beta_{m\ell}^{\ul}-
		\gamuml)}{\Ntx-|\Wm|}+\rho_{\UT}\Ntx \times\nonumber
  \\
&\quad 
				 \sum\limits_{\ell\in\mathcal{K}} \dm^{\OT}
					{\gamumk\beta_{m\ell}^{\ul}}\nonumber
	+\dm^{\SZ}
		\frac{\gamumk}{\Ntx-|\Wm|}+\Ntx\dm^{\OT}
		\gamumk\big),
 \\
  & 
 v_{m i k}^\PZF= \!\dm^{\SZ}
		(\!1\!-\!a_m\!) \frac{\gamumk \beta_{mi}}{\Ntx\!-\!|\Wm|} \!+\!\dm^{\OT}
		\!\Ntx\!
		(1\!-\!a_m)  \gamumk \beta_{mi}.
  \end{align*}

\end{proposition}
%==============
\begin{proof}
The proof follows from~\cite[Lemma B.10]{Emil:Book:2017} by noting that the observed SINR in~\eqref{eq:SINRMA} (the observed SINR in~\eqref{eq:SINRMA_ZF1}) can be written as a generalized Rayleigh quotient with respect to $\boldsymbol{\alpha}_k^\MR$ ($\boldsymbol{\alpha}_k^\PZF$) and thus be solved by a generalized eigenvalue decomposition.
\end{proof}
%=====================

%%==================================
\begin{algorithm}[t]
\caption{Iterative Algorithm to Solve Problem~\eqref{P:SE4} }
\begin{algorithmic}[1]
\label{Alg:iterative}
    \STATE \textbf{Initialize}:  Set super iteration index $i = 0$, choose the initial value of MN-weighting  coefficient  $\ALPHA$. Define the maximum number of iterations $I$.
\REPEAT
    \STATE Compute  power allocation, $\boldsymbol \theta^\star$,  by solving~\eqref{eq:PAoptimization1} using bisection Algorithm~\ref{Alg:PA}. 
\STATE Set $\boldsymbol \theta=\boldsymbol \theta^\star$ and determine optimum MN-weighting  coefficients, $\ALPHA^\star=\{\ALPHA_1^\star, \ldots, \ALPHA_K^\star\}$, through solving the generalized eigenvalue Problem~\eqref{P:LSFD} as in Proposition~\ref{theorem:lsfd}. Set $\ALPHA_k ^\star=\frac{\ALPHA_k^\star}{\Vert\ALPHA_k^\star\Vert}$ $\forall k$.
\STATE Set $i=i+1$ and update $\ALPHA=\ALPHA^\star$.  
\UNTIL Required accuracy or $i = I$
\end{algorithmic}
\end{algorithm}
%\setlength{\textfloatsep}{0.7cm}
%%=======================================

%
Therefore, by combining the two
sub-problems in~\eqref{eq:PAoptimization1} and~\eqref{P:LSFD}, we develop an iterative algorithm  by alternately solving each sub-problem at each iteration, as summarized in \textbf{Algorithm~\ref{Alg:iterative}}.

%------------------------------------------------------------------
\subsection{Complexity and Convergence Analysis}
Here, we quantifying the computational complexity of solving the max-min MSP optimization problem~\eqref{P:SE3}, which involves the proposed greedy MN mode assignment Algorithm~\ref{alg:Grreedy} and the proposed iterative Algorithm~\ref{Alg:iterative} to solve the power allocation and MN-weighting coefficient optimization problem~\eqref{P:SE4}. 
It is easy to show that the complexity of calculating $\mathbb {E}\{X^\CS_k\}$ is on the order of $\mathcal{O}({M^2K})$. Therefore, the complexity of the proposed Algorithm~\ref{alg:Grreedy}  is up to $\frac{M(M+1)}{2}\mathcal{O}({M^2K})$.
Now, we analyze the computational complexity of  {Algorithm~\ref{Alg:iterative}}, which solves the max-min MSP power optimization problem~\eqref{eq:PAoptimization1} by using a bisection method along with solving a sequence of linear feasibility problems based on Algorithm~\ref{Alg:PA} and the generalized eigenvalue problem~\eqref{P:LSFD} at each iteration. 
The total number of  iterations required in Algorithm~\ref{Alg:PA}  is $\log_2(\frac{\zeta_{\max}-\zeta_{\min}}{\epsilon})$. Furthermore, the optimization problem~\eqref{eq:PAoptimization} involves $C_l\triangleq M(K+1)$ linear constraints and $C_v\triangleq MK$ real-valued scalar variables. Therefore,  solving the power allocation by Algorithm~\ref{Alg:PA}  requires a complexity of $\log_2(\frac{\zeta_{\max}-\zeta_{\min}}{\epsilon})\mathcal{O}\big(C_v^2\sqrt{C_l}(C_v+C_l)\big)$. In addition, for the MN-weighting  coefficient design in~\eqref{P:LSFD}, an eigenvalue solver imposes approximately $\mathcal{O}(KM^3)$ flops~\cite{Golub:book}.
    
The convergence of the objective function in the proposed iterative Algorithm~\ref{Alg:iterative} can be charachterized as follows. To solve problem~\eqref{P:SE4}, two sub-problems are alternately
solved so that at each iteration, one set of design parameters is obtained by solving the corresponding sub-problem, while fixing the other set of design variables.
More specifically, at each iteration, the power allocation coefficient set $\boldsymbol{\theta}^\star$ is calculated for the given MN-weighting  coefficient set $\ALPHA$ and then the MN-weighting  coefficient set $\ALPHA ^\star$ is calculated for the given $\boldsymbol{\theta}=\boldsymbol{\theta}^\star$. For the next iteration $\ALPHA$ is updated  as $\ALPHA=\ALPHA ^\star$.  The    power allocation $\boldsymbol{\theta}^\star$ obtained for a given $\ALPHA$ results an MSP greater than
or equal to that of the previous iteration. We also note that the power
allocation solution at each iteration $i$ is also a feasible solution in calculating the power allocation  in the next iteration $i+1$ due to the fact that the MN-weighting  coefficient  in iteration $i+1$ is derived for the  power allocation coefficient given by iteration $i$. Therefore, Algorithm~\ref{Alg:iterative} results in a monotonically increasing sequence of the objectives. 

%%%%%%%%%%%%%%%%%%%%%%%%%%%%%%%%%%%%%%%%%
\begin{algorithm}[!t]
\caption{UT Grouping}
\begin{algorithmic}[1]
\label{alg:UTgrouping} 
\STATE
\textbf{Initialize}: If $\Ntx \geq K + 1$ set $\Wm = \mathcal{K}, \Sm = \emptyset,$ and go to step 15, otherwise set iteration
index $i = 0$.
\STATE Calculate $\Pi^\star[i]= \underset{k\in\mathcal{K}} \min \,\,\mathbb {E}\{X^\CS_k (\Sm, \Wm)\}$ and $k^\star = \underset{k\in\mathcal{K}} {\arg\min} \,\,\mathbb {E}\{X^\CS_k (\Sm, \Wm)\}$
\REPEAT
\FORALL{$m \in \mathcal{M}$}
\STATE Set $\Wm = \Wm'$ and $\Sm = \Sm'$
\IF{$k^\star \notin \Wm$ and $|\Wm| \geq \Ntx-1$}
% \STATE{Update  $\mathcal{A}_{\ul}=\{\mathcal{A}_{\ul}\bigcup m^{\star}\}$ }
% \ELSE
\STATE Set $\Wm' = \{\Wm \cup k^\star\} \setminus \bar{k}_m$ and $\Sm' = \{\Sm \cup \bar{k}_m\} \setminus k^\star$, where $\bar{k}_m$ is the index of UT associated with largest $\betamku$ in $\Wm$.
\ELSIF{$k^\star \notin \Wm$ and $|\Wm| < \Ntx $}  
\STATE $\Wm' = \Wm \cup k^\star$ and $\Sm' = \Sm \setminus k^\star$
\ENDIF
\ENDFOR
\STATE Calculate $\Pi ^\star[i + 1] = \min \mathbb {E}\{X^\CS_k (\Sm', \Wm')\}$ and $k^\star = \underset{k\in\mathcal{K}} {\arg \min}~\mathbb {E}\{X^\CS_k (\Sm', \Wm')\}$
\STATE Set $e = |\Pi^\star [i + 1] - \Pi^\star[i]|$ and $i=i+1$.
\UNTIL{ $e < e_{\min}$ }
\RETURN $\Wm$ and $\Sm$.
\end{algorithmic}
\end{algorithm}
%\setlength{\textfloatsep}{0.7cm}
%---------------------------------------
%============Numerical Results=============
\section{Numerical Results} \label{Sec:numerical}
%%=======================================
In this section, numerical results are presented for studying the performance of the proposed CF-mMIMO surveillance system using the PZF and MR combining  schemes as well as for verifying the
benefit of our max-min MSP optimization framework. We firstly introduce our approach for UT grouping in the PZF combining scheme.
 %===================
\subsection {UT Grouping}
When the number of antennas per MN is sufficiently large, full ZF combining offers excellent performance~\cite{Jiayi:TWC:2021}. Therefore, each MN in observing mode employs the ZF combining scheme for all untrusted links and we set $\Wm = \mathcal{K}$ and $\Sm = \emptyset$ when $\Ntx \geq K+1$. Otherwise,  in each iteration, we assign UT $k$ having minimum MSP to $\Wm$, $\forall m$, until there
is no more improvement in the minimum MSP among the untrusted links, as summarized in \textbf{Algorithm~\ref{alg:UTgrouping}}.
%============Numerical Results=============

\subsection{Simulation Setup and Parameters}
We consider a CF-mMIMO surveillance system, where the MNs and UTs are randomly distributed in an area of  $D \times D$ km${}^2$ having wrapped around edges to reduce the boundary effects. Unless otherwise stated, the size of the network is $D=1$ km. Furthermore, each UR $k$ is randomly located in a circle with radius $150$ m around its corresponding transmitter, UT $k$. Moreover, we set the channel bandwidth to $B=50$ MHz and $\tau_\mathrm{t}=2K$. The maximum transmit power for training  pilot sequences, each MN, and each UT is $250$ mW, $1$ W, and $250$ mW, respectively, while the corresponding normalized maximum transmit powers  $\rho_\mathrm{t}$, ${\rho}_\dl$, and ${\rho}_\UT$ can be calculated upon dividing these powers by the noise power of $\sigma^2_n=-92$ dBm. The large-scale fading coefficient $\beta_{mk}$ is represented by 
%----
\begin{equation} \label{eq:Beta}
\beta_{mk} =10^{\frac{\text{PL}_{mk}}{10}}  10^{\frac {\sigma _{sh} ~y_{mk}}{10}}, 
\end{equation}
%-----------------------------------------------------
where the first term models the path loss, and the second term models the shadow fading  with standard deviation $\sigma_{sh} = 4$ dB, and $y_{mk} \sim \mathcal{CN}(0, 1)$, respectively. Let us denote the distance
between the $m$-th MN and the $k$-th user by $d_{mk}$. Then,   $\text{PL}_{mk}$ (in dB) is calculated as \cite{Hien:cellfree} 
%-------
\begin{align}\label{eq:PL}
&\hspace {-1.8pc}\text {PL}_{mk}=\! \begin{cases} \!-L-35\log_{10}(d_{mk}),\qquad \qquad \quad & d_{mk} > d_{1},\!\! \\ \!-\!L-\!15\!\log _{10}(d_{1})\!-\!\!20 \log_{10}(d_{mk}),~\!&\! d_{0}\! < \!d_{mk}\!\!\le \!\! d_{1},\!\!\\ \!-L-15 \log_{10}(d_{1})-20 \log _{10}(d_{0}),&d_{mk} \le d_{0}, \end{cases} %\!\!\!\!\!\!\!\!\!\!\!\! \\{}\nonumber
\end{align}
%-------
with $L = 46.3 + 33.9 \log_{10}(f)-13.82 \log_{10}(h_{\text{MN}} )-
(1.1 \log_{10}(f)- 0.7)h_{\text{U}} + (1.56 \log_{10}(f)- 0.8)$, where $f$
is the carrier frequency (in MHz), $h_{\text{MN}}$ and $h_{\text{U}}$ denote the MN antenna height (in m) and untrusted user height (in m), respectively. In all examples, we choose $d_0 = 10$ m, $d_1 = 50$ m,  $h_{\text{MN}} = 15$ m and $h_{\text{U}} = 1.65$ m. These parameters resemble those in \cite{Hien:cellfree}.  Similarly,  the large-scale fading coefficient $\beta_{\ell k}$ between  the $\ell$-th UT and  $k$-th UR can be modelled by a  change of indices in~\eqref{eq:Beta} and~\eqref{eq:PL}.

 %%====================================================
 \subsection{Performance Evaluation}
 \begin{figure}[t]
	\centering
	\includegraphics[width=0.45\textwidth]{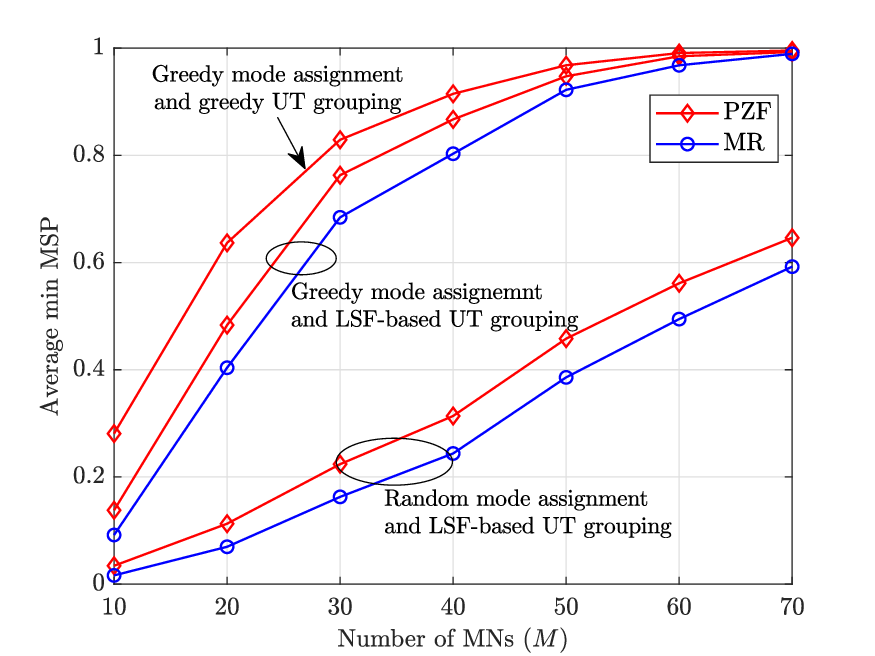}
	\caption{Average minimum MSP with mode assignment Algorithm~\ref{alg:Grreedy}  and UT grouping Algorithm~\ref{alg:UTgrouping}, where  $\Ntx=12$, and $K=20$.}
	\vspace{0em}
	\label{fig:Fig1}
\end{figure}

%%================================================================
\emph{1) Performance of the Proposed Greedy Mode Assignment and Greedy UT Grouping}:
Here, we investigate the performance of the proposed greedy mode assignment in Algorithm~\ref{alg:Grreedy} and greedy UT grouping Algorithm~\ref{alg:UTgrouping} for PZF and MR combining schemes. We benchmark $1)$ random mode assignment, $2)$ UT grouping based on the value of large-scale fading coefficient $\betamku$ (LSF-based UT grouping),  so that when $\Ntx > K$ all UTs are assigned to $\Wm$ for ZF combining, $\forall m$, and when $\Ntx \leq K$ at MN $m$ a UT with the smallest value of $\betamku$ is assigned into the $\Wm$ group for ZF combining and the remaining  UTs are assigned into the $\Sm$ group for MR combining.
Figure~\ref{fig:Fig1} illustrates the minimum MSP  achieved by the  CF-mMIMO surveillance system for different number of MNs $M$. In this initial evaluation,
the setup consists in $D = 1$ km, $\Ntx = 12$, and $K = 20$.
 Our results verify the advantage of the proposed greedy mode assignment over random mode assignment.  More specifically, when $M=30$, greedy mode assignment provides performance gains of around $245\%$ and $325\%$ with respect to random mode assignment for the system relying on PZF  combining and MR combining, respectively. This remarkable performance gain  verifies the importance of an adequate mode selection in terms of monitoring performance in our  CF-mMIMO surveillance system. Additionally, compared to the LSF-based grouping scheme, our proposed UT grouping  provides up to an additional $100\%$  improvement in terms of MSP. This is reasonable  because,   PZF combining employing our proposed UT grouping can achieve an attractive  balance between mitigating the interference   and increasing the array gain. 
  In the next figures, we present results for the scenarios associated with greedy mode assignment and greedy UT grouping.
%----------------------------------------

 \emph{2)  Performance of the Proposed Max-Min MSP}: Now, we examine the efficiency of proposed Max-Min MSP using power control and MN-weighting  coefficient control provided by Algorithm~\ref{Alg:iterative} for the PZF and MR combining  schemes. Our numerical results (not shown here) demonstrated that Algorithm~\ref{Alg:iterative} converges  quickly, and hence in what follows we set the maximum number of  iterations to $I = 2$ for Algorithm~\ref{Alg:iterative}.
Figure~\ref{fig:Fig2} presents the minimum MSP of the CF-mMIMO surveillance system for  different numbers of antennas per MN for systems having the
same total numbers of service antennas, i.e., $\Ntx_{\mathtt{tot}}=\Ntx M = 240$, but different
number of MNs.  We investigate three cases: \emph{case-$1$)} equal power allocation and equal MN-weighting  coefficient control, \emph{case-$2$)} proposed power control
but no optimal MN-weighting  coefficient ($\alpha_{mk} = 1$ and $\theta_{mk}$ is calculated from Algorithm~\ref{Alg:PA}), \emph{case-$3$)}  power control  and optimal MN-weighting 
coefficient control Algorithm~\ref{Alg:iterative}.  
The main observations that follow from these simulations are
as follows:
%------
\begin{itemize}
\item The max-min MSP power control and MN-weighting  coefficient control enhance the system performance significantly for both the PZF and MR combining schemes. In particular, for the PZF combining scheme, compared to the \emph{case-$1$}, i.e., equal power control and equal MN-weighting  coefficient control, the power control   provides a performance gain of up to $35 \%$, while the power control together with the MN-weighting  coefficient control  can provide a performance gain of up to $43\%$. This highlights the advantage of our proposed solution, which becomes more pronounced  for the PZF combining scheme.

%--------------------
\begin{figure}[t]
	\centering
	\includegraphics[width=0.45\textwidth]{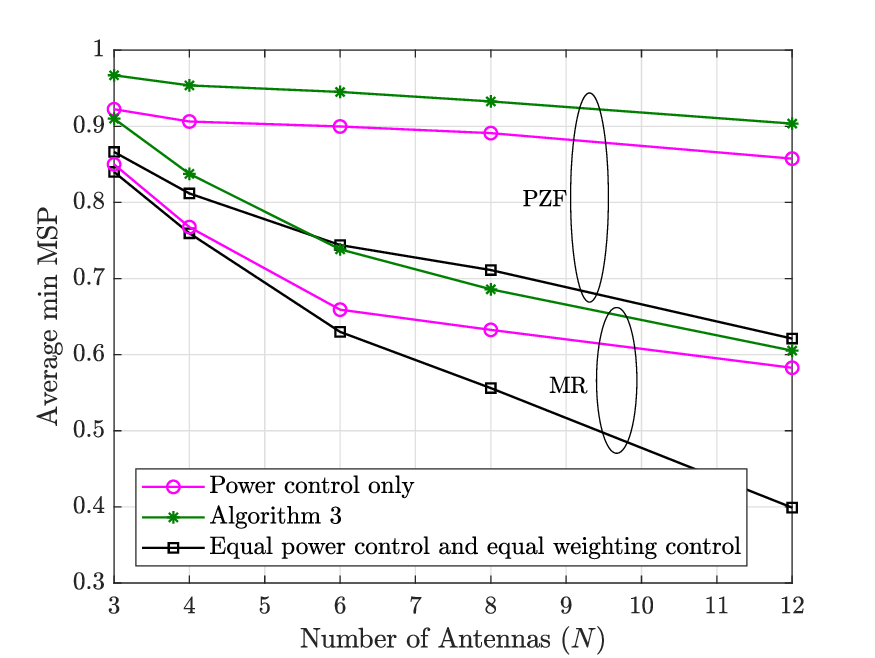}
	\caption{Average minimum MSP with max-min MSP optimization~\eqref{P:SE3}, where $K=20$ and $\Ntx_{\mathtt{tot}}=240$.}
	\vspace{0.5em}
	\label{fig:Fig2}
\end{figure}
%---------------------

\item The monitoring performance gap between the MR combining and the PZF combining
is quite significant. In particular, when $\Ntx=12$, applying PZF combining leads to $50\%$ improvement in terms of MSP with respect to the MR combining scheme. The reason is two-fold: Firstly, the ability of the PZF combining to cancel the cross-link interference; Secondly,   the proposed power control and MN-weighting  coefficient control along with the UT grouping scheme can  notably enhance the monitoring performance of weak UTs. We also note that, the performance gap between PZF and MR combining schemes increases upon increasing $\Ntx$. The intuitive reason is that  for a fixed total number of  antennas, when the number of antennas per MN increases, the number of MN reduces. For a low number of MNs, the cross-link  interference becomes dominant, which significantly degrades the overall performance of the system relying on MR combining.

\item When $\Ntx$ increases, the performance of the MR and PZF combining schemes deteriorates. This is due to the fact that increasing $\Ntx$  and accordingly decreasing $M$ has two effects on the
MSP, namely, (i) increases the
diversity and array gains (a positive effect), and (ii) reduces the macro-diversity gain and increases the path loss due to an increase in the relative distance
between the MNs and the untrusted pairs (a negative effect). The
latter effect becomes dominant, which leads to a degradation in the monitoring performance.
\end{itemize}
%-----------------
%%====================================================
\begin{figure}[t]
	\centering
	\includegraphics[width=0.45\textwidth]{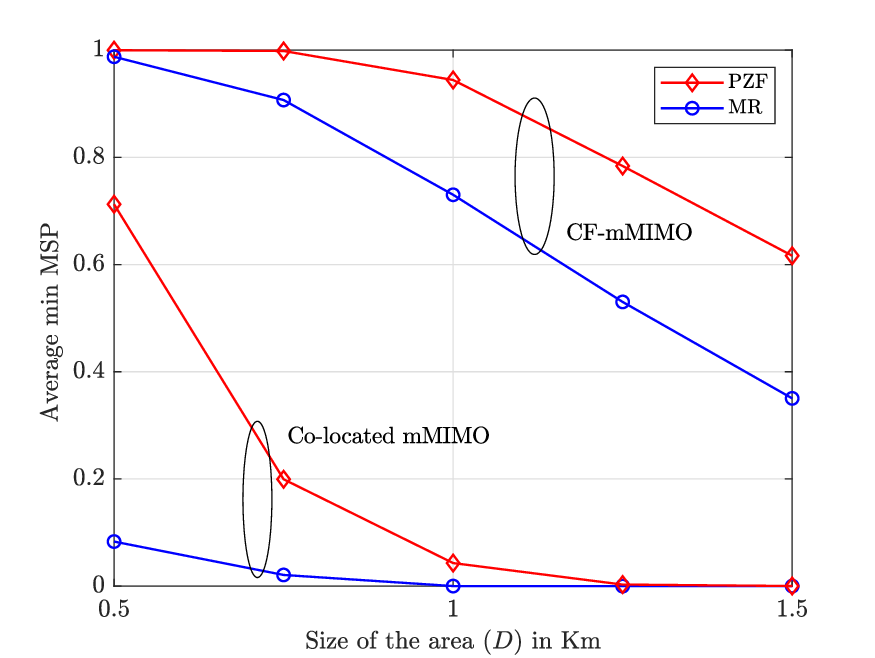}
	\vspace{-0.2em}
	\caption{Average minimum MSP  versus the size of the area, $D$, where $M=40$, $\Ntx=6$, and $K=20$.}
	\vspace{0.5em}
	\label{fig:Fig3}
\end{figure}
%%======================================
 \emph{3)  CF-mMIMO  versus Co-located Surveillance System}:
 Now, we compare the  MSP of the CF-mMIMO against that of a co-located FD massive MIMO system. The co-located FD massive MIMO surveillance system can be considered as a special case of the CF-mMIMO system,  where all $M$ MNs are co-located as an antenna array, which simultaneously performs observation and jamming  at the same frequency. 
 Therefore, the effective SINR of the untrusted link and the effective SINR for observing at the CPU can be obtained by setting   $\betamku=\beta_{ik}^\ul=\beta_k^\ul, \betamkd=\beta_{ik}^\dl=\beta_k^\dl$, $\gamumk=\gamma_{ik}^\ul=\gamma_k^\ul, \gamdmk=\gamma_{ik}^\dl=\gamma_k^\dl$,  $\beta_{mi}=\sigma^2_{\SI}, \forall m, i, k$ in Propositions~\ref{Theorem:SE:SUS},~\ref{Theorem:SE:CPU}, and~\ref{Theorem:SEZF:CPU}, respectively. Here, $\sigma^2_{\SI}$ reflects the strength of the residual self interference after employing self-interference suppression techniques~\cite{Riihonen:TSP:2011}. Recall that in our CF-mMIMO surveillance system all MNs operate in half-duplex mode, hence there is no self interference at each MN. 

For fair comparison with the CF-mMIMO system, the  co-located system deploys the same total number of antennas, i.e., $\Ncl=\frac{\Ntx_{\mathtt{tot}}}{2}=\frac{\Ntx M}{2}$ antennas are used for observing, while 
${\Ncl}$ antennas are used for jamming, which is termed as ``an antenna-preservation" condition~\cite{Himal:TWC:2014}. Accordingly, the effective SINR  of the untrusted link $k$ for FD co-located massive MIMO systems can be written as 
%================
\begin{align}\label{eq:SINRSR_colocated}
    \SINR_{\UR,k}^\CL  (\boldsymbol \theta) = \frac{\rho_{\UT}|  h_{kk}|^2}{\xi_k(\boldsymbol{\theta})},
\end{align}
\vspace{-0.5em}
%----------------------
where
\vspace{-0.2em}
\begin{align}
\xi_k(\boldsymbol{\theta})
=& \rho_{\UT}\!\!\!\sum_{\ell\in \mathcal{K}\setminus k} \!\!\betakldu\!+{\rho_\dl}\Ncl
\sum_{k'\in\mathcal{K}}\!\!
 \theta_{k'}\beta_{k}^\dl\gamma_{k'}^\dl  
\!\!+\!{\rho_\dl}\Ncl^2\!{\theta_{k}}(\gamma_{k}^\dl)^2\!+\!1.\nonumber
 \end{align}
%-----------------------
Additionally, the received SINR of the $k$-th untrusted link at the CPU for MR combining in our FD co-located massive MIMO system is given by
%-----
\vspace{-0.5em}
\begin{align}\label{eq:SINRMA_colocated}
&\SINR_{\ul,k}^{\CL,\MR}(\boldsymbol{\theta})=\frac{
	\Ncl  \rho_{\UT} (\gamma_k^\ul)^2
	}
	{	 \mu_{k}^{\CL,\MR}+
		\rho_\dl \Ncl
		\sum\nolimits_{\ell\in\mathcal{K}}
		\!
		\theta_{\ell}  \gamma_{\ell}^{\dl} \varrho_{k}^{\CL,\MR}
},
\end{align}
with 
$\mu_{k}^{\CL,\MR}= \rho_{\UT}
		 \sum\nolimits\nolimits_{\ell\in\mathcal{K}}\beta_{\ell}^{\ul}
		\gamma_{k}^{\ul}
	+\gamma_k^\ul$, and
$\varrho_{k}^{\CL,\MR}= \gamma_k^\ul \sigma^2_{\SI}$,
 %-----------------
while the received SINR for the $k$-th untrusted link for full ZF combining and  $\Ncl \geq K+1$ is given by
%-----------------
%-----------------
\begin{align}\label{eq:SINRMA_ZF_colocated}
&\SINR_{\ul,k}^{\CL,\ZF}(\boldsymbol{\theta})=
\frac{
	\rho_{\UT} ({\gamma_{k}^\ul})^2
	}
	{\mu_k^{\CL,\ZF}+
		\rho_\dl\Ncl
		\sum\nolimits_{\ell\in\mathcal{K}}
		\theta_\ell  {\gamma_{\ell}^{\dl}}\varrho_{k}^{\CL,\ZF}
},
\end{align}
with 
$\mu_k^{\CL,\ZF}=\rho_{\UT}
		 \sum\nolimits_{\ell\in\mathcal{K}}\!
	\frac{\gamma_k^\ul(\beta_{\ell}^{\ul}-
		\gamma_\ell^\ul)}{\Ncl-K}
+		\frac{\gamma_{k}^\ul}{\Ncl-K}$ and
$ \varrho_{k}^{\CL,\ZF}= \frac{\gamma_{k}^\ul \sigma^2_{\SI}}{\Ncl-K}$. 

%----------------------

 %%====================================================
\begin{figure}[t]
	\centering
	\includegraphics[width=0.45\textwidth]{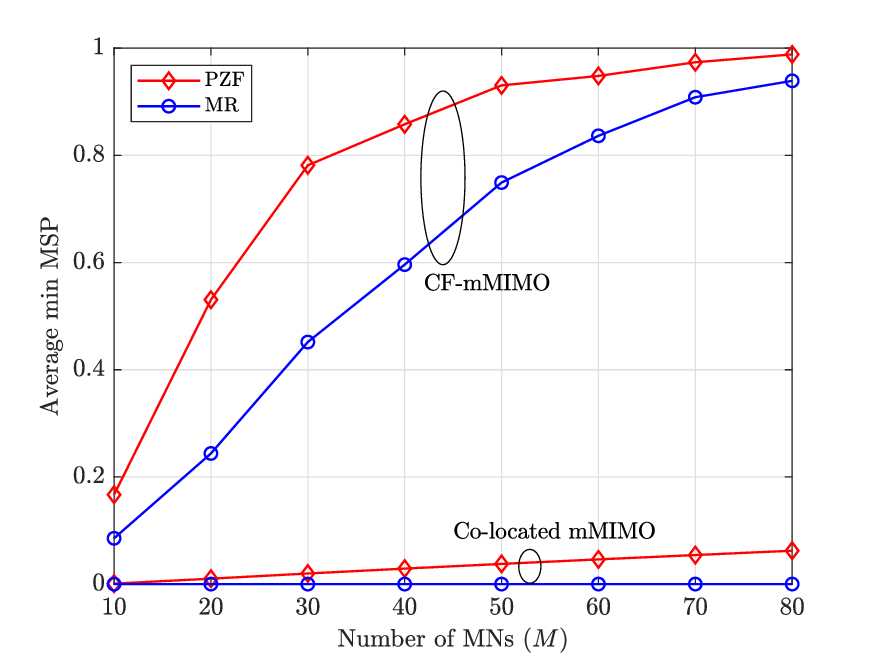}
	\vspace{-0.2em}
	\caption{Average minimum MSP  versus number of MNs, $M$, where  $\Ntx=4$, and $K=20$.}
	\vspace{0em}
	\label{fig:Fig4}
\end{figure}
%%======================================

For co-located massive MIMO systems, unless otherwise stated, we assume that the residual self interference after employing a self-interference suppression technique is $\sigma^2_{\SI} /\sigma_n^2=30$ dB\footnote{The strength of the  $\sigma^2_{\SI} /\sigma_n^2$ after employing employing self-interference suppression technique is typically in the range of $30$ dB to $100$ dB~\cite{Riihonen:TSP:2011}. Therefore, the performance of co-located massive MIMO with $\sigma^2_{\SI} /\sigma_n^2=30$ dB can be regarded as an upper bound.}. In addition, we  adopt a similar power control principle as in Algorithm~\ref{Alg:PA}.
Figure~\ref{fig:Fig3} shows the minimum MSP versus the size of the area, $D$.  It can be observed that the CF-mMIMO  surveillance system  significantly outperforms its co-located  massive MIMO counterpart. As expected, the relative  performance  gap between CF-mMIMO  and co-located surveillance systems dramatically escalates with the size of the area $D$. For example, for $D=0.75$ km, the CF-mMIMO system provides around $5$-fold improvement in the minimum MSP performance  over the co-located system, while the improvement reaches a $40$-fold value for  $D=1$ km. This highlights the effectiveness of our  optimized CF-mMIMO surveillance scheme for proactive monitoring systems.  The reason is   the capability of the  CF-mMIMO to surround each UT and each UR  by relying on MNs operating in observing and jamming mode, respectively. Additionally, in contrast to our CF-mMIMO with distributed MNs, a co-located massive MIMO  suffers from excessive self-interference due to the short distance between the transmit and receive antennas of a single large MN.

 \emph{4) Effect of the Number of MNs}:
 Figure~\ref{fig:Fig4} shows the minimum MSP versus the number of MNs, $M$. For our CF-mMIMO system using the MR and PZF combining schemes, when  the number of MNs increases, the macro-diversity gain increases, and hence the MSP enhances.  Upon  increasing $M$, for co-located massive MIMO  the number of transmit and receive antennas increases, which results in a higher array gain and monitoring performance.  In particular, we can see that the surveillance systems can benefit much more from  the higher macro-diversity   gain in a CF-mMIMO network, rather than from the higher array gain attained in co-located networks.

The results shown both in Fig.~\ref{fig:Fig3} as well as in Fig.~\ref{fig:Fig4} clearly suggest that having  a high degree of macro diversity and low path loss are crucial for offering a high MSP and corroborate that CF-mMIMO is well-suited for the surveillance  of  networks  in wide areas.
% 

% %%====================================================
\begin{figure}[t]
	\centering
	\includegraphics[width=0.45\textwidth]{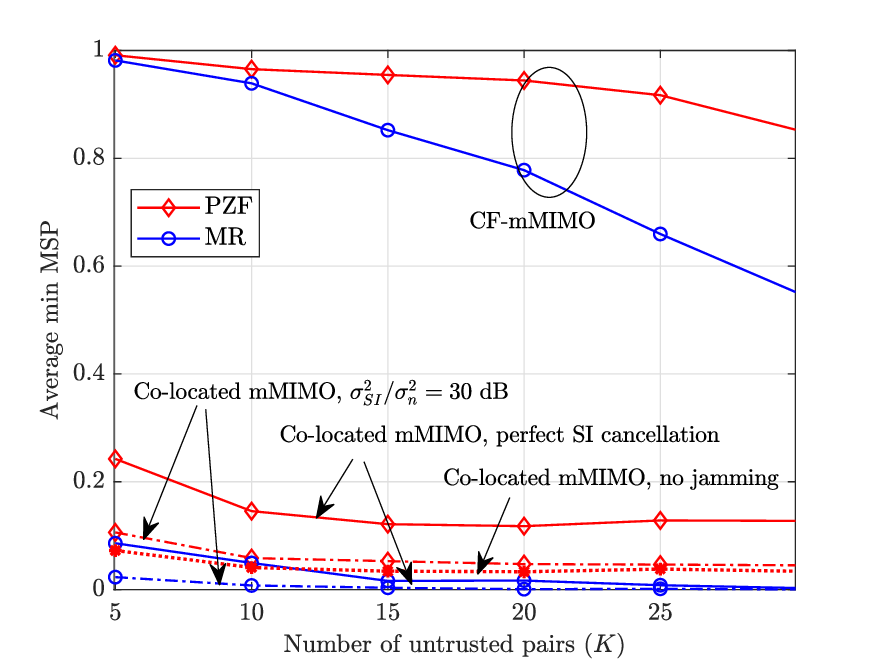}
	\vspace{-0.2em}
	\caption{Average minimum MSP  versus number of untrusted links, $K$, where $M=40$ and $\Ntx=6$.}
	\vspace{0em}
	\label{fig:Fig5}
\end{figure}
% %%======================================
  
  \emph{5) Effect of the Number of Untrusted Links}:  
Next, in Fig.~\ref{fig:Fig5}, we investigate the impact of the number of untrusted links on the MSP performance of both CF-mMIMO and of co-located massive MIMO systems. Herein, we also consider co-located massive MIMO having perfect SI cancellation and half-duplex  co-located massive MIMO with no jamming. We observe that upon increasing $K$, the monitoring performance of all schemes deteriorates. Nevertheless, the CF-mMIMO system using the PZF combining scheme  still yields excellent  MSP compared to the other schemes. We also note that for  small values of $K$, MR combining provides a better
performance/implementation complexity trade-off compared to its PZF counterpart. However, increasing the number of untrusted pairs results in stronger cross-link interference. Therefore, PZF combining having the ability to cancel  the cross-link interference is  undoubtedly a better choice.   Interestingly, we can observe that even under the idealized assumption of having perfect SI cancellation in co-located massive MIMO, CF-mMIMO surveillance still significantly outperforms  the co-located massive MIMO. This result shows that the proposed CF-mMIMO surveillance system, relying on the proposed power control, MN-weighting  coefficient control, and suitable mode assignment,  yields an  impressive monitoring performance in   multiple untrusted pair scenarios.

\section{Conclusions}
%=======================================
We have  developed a CF-mMIMO surveillance system for monitoring  multiple distributed untrusted pairs and analyzed the  performance of both MR and PZF combining schemes. We proposed a new long-term-based  optimization technique of designing the MN mode assignment, power control for the MNs that are in jamming mode, and MN-weighting 
coefficients to maximize the min MSP across all untrusted pairs  under practical transmit power constraints.
We showed that our  CF-mMIMO surveillance systems  provide significant  monitoring gains over conventional co-located massive MIMO, even for relatively small number of MNs. In particular, the minimum MSP of CF-mMIMO surveillance is an order of magnitude higher than that of the co-located massive MIMO system, when the untrusted pairs are spread out over a large area. The results also show that with different network setups, PZF combining provides the highest MSP, while for small values of the number of untrusted pairs and the size of the area,  MR combining constitutes a beneficial choice.

%%%%%%%%%%%%%%%%%%%%%%%%%%%%%%%%%%%%%%%%%%%%%%%%%%%%%%
\vspace{-0.5em}
\appendices
%%%%%%%%%%%%%%%%%%%%%%%%%%%%%%%%%%%%%%%%%%%%%%%%%%%%%%%
\section{Proof of Proposition~\ref{Theorem:SE:SUS}}~\label{ProofTheorem:SE:SUS}
To derive a closed-form expression for  the effective SINR  of the untrusted link, we have to calculate $\xi_k(\qa,\boldsymbol{\theta})=\mathbb{E}\Big\{\big|\tilde w_{k}^{\UR}\big|^2\Big\}$. Let us denote   the jamming channel
estimation error by $\boldsymbol{\varepsilon}_{mk}^{\mathrm{J}}=\gmkd-\hgmkd$. Therefore, we have
% %============
 \begin{align}\label{eq:SINRD1}
&\xi_k(\qa,\boldsymbol{\theta})=\mathbb{E}\Big\{\big|\tilde w_{k}^{\UR}\big|^2\Big\}=
 \rho_{\UT}\sum_{\ell\in \mathcal{K}\setminus k} \betakldu+{\rho_\dl} \mathcal{I}+1, 
  \end{align}
 %----
  where
  %----
  \vspace{-0.4em}
 \begin{align*}
 \mathcal{I}\triangleq & \sum_{k'\in\mathcal{K} } \mathbb{E} \Big\{\Big|
\sum_{m \in \mathcal{M}}
 a_m\sqrt{\theta_{mk'}}
(\gmkd)^T\left(\hgmkpd\right)^*
\Big|^2\Big\}.
  \end{align*} 
%%%%%%%%%%%%%%%%%%%%%%%
To calculate $\mathcal{I}$, owing to the fact that the variance of a sum of
independent RVs is equal to the sum of the variances, we have
%===============
 \begin{align}\label{eq:SINRD2}
 &\mathcal{I}\stackrel{(a)}{=}\!
\sum_{k'\in\mathcal{K}\setminus k} \!\!\mathbb{E} \Big\{\big| \!
 \sum_{m \in \mathcal{M}}\!
\! a_m\sqrt{\theta_{mk'}}\left(\gmkd\right)^T\left(\hgmkpd\right)^*\big|^2\Big\}
\nonumber\\
 &\qquad
 + \!\mathbb{E}\Big\{
\Big|\sum_{m \in \mathcal{M}}\!\!
a_m \sqrt{\theta_{mk}}
\left(\boldsymbol{\varepsilon}_{mk}^{\mathrm{J}}\!+\!\hgmkd\right)^T\left(\hgmkd\right)^*\Big|^2\Big\}\nonumber
 \\
&\!\!\stackrel{(b)}{=}\!\!\! 
\sum_{k'\in\mathcal{K}\setminus k} 
 \sum_{m \in \mathcal{M}}\!\!
\!a_m {\theta_{mk'}}\mathbb{E}\! \Big\{\!\left(\gmkd\right)^T\!\!\mathbb{E} \Big\{\!\left(\hgmkpd\right)^*\!\!\left(\hgmkpd\right)^T\!\!\Big\}\left(\gmkd\right)^*\!\!\Big\}\nonumber
 \\
&\hspace{2em}+\sum_{m \in \mathcal{M}}a_m\theta_{mk}\left(\mathbb{E} \Big\{\big\Vert\hgmkd\big\Vert^4\Big\}+\mathbb{E} \Big\{\big|{(\boldsymbol{\varepsilon}}_{mk}^{\mathrm{J}})^T(\hgmkd)^*\big|^2\!\Big\}\!\right)\nonumber\\
&\hspace{2em}+\sum_{m \in \mathcal{M}}\sum_{n \in \mathcal{M}\setminus m}a_m a_n\sqrt{\theta_{mk}\theta_{nk}} \mathbb{E} \Big\{\Vert\hgmkd\Vert^2\!\Big\}
\mathbb{E} \Big\{\Vert\hgnkd\Vert^2\!\Big\}\nonumber
 \\
&\stackrel{(c)}{=} 
\Ntx\!\!
 \sum_{k'\in\mathcal{K} \setminus k}\sum_{m \in \mathcal{M}}
 a_m\theta_{mk'}\betamkd\gamdmkp  \nonumber
 + \!\Ntx\!\!\sum_{m \in \mathcal{M}}\!\!a_m\theta_{mk}\gamdmk(\Ntx \times
\nonumber
\\
&\qquad\gamdmk+\betamkd)+\Ntx^2\sum_{m \in \mathcal{M}}\sum_{n \in \mathcal{M}\setminus m}a_m a_n\sqrt{\theta_{mk}\theta_{nk}} \gamdmk \gamdnk\nonumber\\
&=\!
\Ntx\!\!
 \sum_{k'\in\mathcal{K}}\sum_{m \in \mathcal{M}}
 a_m\theta_{mk'}\betamkd\gamdmkp  
 %\nonumber
%  \\
% &\hspace{3em} 
 +\Ntx^2\Big(\!\sum_{m \in \mathcal{M}}\!\!a_m \sqrt{\theta_{mk}} \gamdmk\! \Big)^2\!\!,
  \end{align}
  % %  %------------
where $(a)$  follows from the fact that $\hgmkpd$ has zero mean and it is independent of $\gmkd$ and $\hgmkd$, $(b)$ follows from the fact that $\boldsymbol{\varepsilon}_{mk}^{\mathrm{J}}$   is independent of $\hgmkpd$ and it is a zero-mean RV
 and $(c)$ follows from  the fact that $\mathbb{E} \big\{\Vert\hgmkd\Vert^4\big\}=\Ntx(\Ntx+1)(\gamdmk)^2$.
Substituting~\eqref{eq:SINRD2} into~\eqref{eq:SINRD1} completes the proof.
%===============================================================
\section{Proof of Proposition~\ref{Theorem:SE:CPU}}~\label{ProofTheorem:SE:CPU}
Let us denote   the observing channel estimation error, which is independent of $\hat \qg_{mk }^\ul$ and  zero-mean RV, by $\boldsymbol{\varepsilon}_{mk}^\ul=\qg_{mk }^\ul- \hat \qg_{mk }^\ul$. According to $\vmkm=\hgmku$ and by exploiting the independence between the channel
estimation errors and the estimates, we have
%------------------
\begin{align} \label{eq:DS}
\mathbb{E}\big\{\text{DS}_{k} ^\MR\} &\!=\! \sqrt {{\rho _\UT}}  \sum\limits_{m \in {\mathcal{M}}} \!\!\alphmk(1\!-\!a_m)  \mathbb{E}\left\{{(\hgmku)^\dag}(\hgmku \!+ \!{\boldsymbol{\varepsilon}}_{mk }^\ul)\right\}\nonumber\\ 
& = \sqrt {{\rho _\UT}} \sum\limits_{m \in {\mathcal{M}}} \alphmk (1-a_m) \Nrx\gamumk.
\end{align}
%============
Additionally,  $\mathbb{E}\big\{\vert\text{BU}_k ^\MR \vert^2\big\}$ can be written as
%-------------
\begin{align} \label{eq:BI}
& \mathbb{E}\big\{\vert\text{BU}_k ^\MR \vert^2\big\}
\nonumber
\\
&
= \!{\rho_\UT}\sum\limits_{m \in {\mathcal{M}}} \alphmk^2(1\!-\!a_m)  \mathbb{E}\Big\{ |{(\hgmku)^\dag}\gmku \!-\! \mathbb{E}\{ {(\hgmku)^\dag}\gmku\} {|^2}\!\Big\} \nonumber
\\ 
 &= \!\rho_\UT\!\!\!\sum\limits_{m \in {\mathcal{M}}}
 \alphmk^2(1-a_m)  \Big(\!{ \mathbb{E}}\Big\{ \big|{(\hgmku)^\dag}\boldsymbol{\varepsilon}_{mk}^\ul \!+\! \Vert\hgmku{\Vert^2}{\big|^2}\!\Big\} \nonumber
 \\
&\hspace{10em}
\!-\! \Nrx^2(\gamumk)^2\!\Big)\nonumber
\\
& \stackrel{(a)}{=} \rho _\UT \sum\limits_{m \in {\mathcal{M}}}\alphmk^2(1-a_m)  \big(\mathbb{E}\big\{ |{(\hgmku)^\dag}\boldsymbol{\varepsilon}_{mk}^\ul{|^2}\big\}   \nonumber
\\
&\hspace{10em}
+\mathbb{E}\big\{ \Vert\hgmku{\Vert^4}\big\}- \Nrx^2(\gamumk)^2\big) \nonumber
\\ 
&\stackrel{(b)}{=} \rho _\UT \sum\limits_{m \in {\mathcal{M}}}\alphmk^2 (1-a_m)  \Nrx\beta _{mk }^\ul\gamumk, 
\end{align}
%===========
where we have exploited that: in $(a) $ $\boldsymbol{\varepsilon}_{mk}^\ul$  is independent of $\hat \qg_{mk }^\ul$ and a zero-mean RV; in $(b)$ $\mathbb{E} \big\{\Vert\hgmku\Vert^4\big\}=\Ntx(\Ntx+1)(\gamumk)^2$.

%----------------
By exploiting the fact that $\hgmku$ is independent of $\gmlu$ for $k \neq \ell$, while $\hgmku$,  ${\qF_{mi}}$, and ${\mathbf{\hat g}}_{i\ell}^\dl$ are independent, we can formulate  $\mathbb{E}\{\vert {\text{UI}}_{\ell k} ^\MR\vert^2\}$ and $\mathbb{E}\{\vert\text{MI}_{\ell k} ^\MR\vert^2\}$, respectively, as
%===============
\vspace{0em}
\begin{align} \label{eq:SI}
&\mathbb{E}\big\{\vert {\text{UI}}_{\ell k} ^\MR\vert^2\big\}= {\rho _\UT} {\sum\limits_{m \in {\mathcal{M}}}  }\alphmk^2(1-a_m) \Ntx\gamumk\beta _{m\ell}^\ul,\\ 
%\end{align}
% %==============
% and
% %===========
% \vspace{0.2em}
%\begin{align} 
&\mathbb{E}\big\{\vert\text{MI}_{\ell k} ^\MR\vert^2\big\} \nonumber
\\
&={\rho _\dl} \sum\limits_{m \in {\mathcal{M}}} \sum\limits_{i \in {\mathcal{M}}}\alphmk^2 (1-a_m) a_{i}\theta_{i\ell} 
\mathbb{E}   \{ |{({\mathbf{\hat g}}_{mk }^\ul)^H}{{\mathbf{F}}_{mi}}{({\mathbf{\hat g}}_{i\ell}^\dl)^{\ast}}{|^2}\} \nonumber
\\
&= {\rho _\dl} {\sum\limits_{m \in {\mathcal{M}}} {\sum\limits_{i \in {\mathcal{M}}}  } } \alphmk^2\Ntx^2		(1-a_m) a_{i}\theta_{i\ell}   \gamumk \beta_{mi} \gamma_{i\ell}^{\dl}\label{eq:RI}.
\end{align}
%==============
 The substitution of~\eqref{eq:DS},~\eqref{eq:BI},~\eqref{eq:SI}, and~\eqref{eq:RI} into~\eqref{eq:SINRST} and inserting  $\mathbb{E}\big\{ \vert\text{AN}_k ^\MR\vert^2\big\}=\! \Nrx\sum\nolimits_{m \in {\mathcal{M}}}\alphmk^2(1-a_m) \gamumk$, yields~\eqref{eq:SINRMA}.
%===========================================================
%=========================================================
\section{Proof of Proposition~\ref{Theorem:SEZF:CPU}}~\label{ProofTheorem:SEZF:CPU}
According to~\eqref{eq:ZF_prec2} and due to the  fact that ${\boldsymbol{\varepsilon}}_{mk }^\ul$ has zero mean and is
independent of $\hgmku$, $ \mathbb{E}\{\text{DS}_k ^\PZF\}$ in the numerator of~\eqref{eq:SINRST} can be calculated as
%------------------
\begin{align} \label{eq:DS_zf}
& \mathbb{E}\{\text{DS}_k ^\PZF\} \!=\! \! \sqrt {{\rho _\UT}} \Big(\!\!\sum\limits_{m \in \Zk} \!\!\alphmk(1\!\!-\!a_m)  \mathbb{E}\left\{{(\vmkz)^\dag}(\hgmku \!+\!{\boldsymbol{\varepsilon}}_{mk }^\ul)\right\} \nonumber\\
 &\hspace{3em}+\sum\limits_{m \in \Tk} \alphmk(1-a_m)  \mathbb{E}\left\{{(\vmkm)^\dag}(\hgmku + {\boldsymbol{\varepsilon}}_{mk }^\ul)\right\}\Big)\nonumber\\ 
& = \sqrt {{\rho _\UT}}\Big(\sum\limits_{m \in \Zk} \alphmk  (1-a_m) \gamumk \nonumber\\
&\hspace{3em}+ \Ntx\sum\limits_{m \in\Tk} \alphmk  (1-a_m) \gamumk\Big).
\end{align}
%============
Then,  $\mathbb{E}\big\{\vert{\text{BU}}_k ^\PZF\vert^2\big\}$ can be written as
%-------------
\begin{align} \label{eq:BI_zf1}
& \mathbb{E}\big\{\vert\text{BU}_k ^\PZF \vert^2\big\}
= {\rho_\UT} \mathbb{E}\Big\{ \big \vert\sum\limits_{m \in \Zk} \alphmk(1\!-\!a_m)  {(\vmkz)^\dag}\gmku+\nonumber
\\
&\sum\limits_{m \in \Tk} \alphmk(1\!-\!a_m)  {(\vmkm)^\dag}\gmku\big|^2\Big\} 
- {\rho_\UT} \big \vert \mathbb{E} \big\{ \sum\limits_{m \in\Zk} \alphmk\nonumber\\
&(1\!-\!a_m)  {(\vmkz)^\dag}\gmku+\sum\limits_{m \in\Tk} \alphmk(1\!-\!a_m)  {(\vmkz)^\dag}\gmku\big\}\big|^2 \nonumber
\\ 
 &= \rho_\UT\mathcal{I}_2- |\mathbb{E}\big\{\text{DS}_k ^\PZF\}|^2,
\end{align}
%=====
where
%======
\vspace{-0.5em}
\begin{align}
\mathcal{I}_2 =& \mathbb{E}\Big\{ \big \vert\sum\nolimits_{m \in \Zk} \alphmk(1\!-\!a_m)  {(\vmkz)^\dag}\gmku+\nonumber
\\
&\sum\nolimits_{m \in \Tk} \alphmk(1\!-\!a_m)  {(\vmkm)^\dag}\gmku\big|^2\Big\}.
\end{align}
%===========
By applying 
%====
\begin{align}
\mathbb{E}\Big \{ \big\Vert\vmkz\big\Vert^2\Big\}&=(\gamumk)^2\mathbb{E}\Big \{ \big\Vert\hGm\big((\hGm)^\dag\hGm\big)^{-1}\qe_k\big\Vert^2\Big\}\nonumber
\\
&=\frac{\gamumk}{\Ntx-|\Wm|},
\end{align}
%====
which follows from~\cite[Lemma 2.10]{tulino04}, we have
%===
\begin{align} \label{eq:BI_zf2}
& \mathcal{I}_2
\stackrel{(a)}{=} \big(\sum_{m \in \Zk}\alphmk(1-a_m)\gamumk\big)^2\!+\!\sum_{m\in\Zk}\alphmk^2(1-a_m) \times\nonumber
\\    
&\frac{\gamumk(\betamku\!\!-\!\gamumk)}{\Ntx-|\Wm|}+\mathbb{E}\Big\{\big\vert\!\!
 \sum\limits_{m \in \Tk} \alphmk(1\!-\!a_m)  (\vmkm)^\dag\gmku\big\}\big\vert^2\Big\}\nonumber
 \\
 &+2\big(\!\sum_{m \in \Zk}\alphmk(1-a_m)\gamumk\Big)\Big(\Ntx\sum_{m \in \Tk}\!\!\alphmk(1-a_m)\gamumk\big).
\end{align}
%=====
% where in $(a)$   the cross expectations vanish as $\boldsymbol{\varepsilon}_{ik}$ is  independent of $\vmku$ and has zero mean.
Also, the third term of $\mathcal{I}_2$ can be calculated as
%----------------
\begin{align} \label{eq:BI_zf2_T2}
&\mathbb{E}\Big\{\big\vert
 \sum\nolimits_{m \in \Tk} \alphmk(1\!-\!a_m)  (\vmkm)^\dag\gmku\big\}\big\vert^2\Big\}\nonumber\\
&~~=
 \sum\nolimits_{m \in \Tk} \alphmk^2(1\!-\!a_m)  \mathbb{E}\big\{\vert(\vmkm)^\dag\gmku\vert^2\big\}\nonumber
 \\
&\qquad + \Big\vert
 \sum\nolimits_{m \in \Tk} \alphmk(1\!-\!a_m)  \mathbb{E}\big\{(\vmkm)^\dag\gmku\big\}\Big\vert^2\nonumber
 \\
 &\qquad-
  \sum\nolimits_{m \in \Tk}(1\!-\!a_m)  \big\vert
 \alphmk \mathbb{E}\big\{(\vmkm)^\dag\gmku\big\}\big\vert^2\nonumber\\
 &~~ =\sum\nolimits_{m \in \Tk} \alphmk^2(1\!-\!a_m)\Ntx \betamku \gamumk \nonumber\\
 &\qquad+  \Big(\sum\nolimits_{m \in \Tk} \alphmk(1\!-\!a_m)\Ntx\gamumk\Big)^2.
\end{align}
%===========
Substituting~\eqref{eq:BI_zf2_T2} into~\eqref{eq:BI_zf2} and then~\eqref{eq:BI_zf2} and~\eqref{eq:DS_zf} into~\eqref{eq:BI_zf1} yields
%===
\vspace{-0.5em}
\begin{align} \label{eq:BI_zf}
\mathbb{E}\big\{\vert \text{BU}_k ^\PZF\vert^2\}=&\rho_\UT\sum_{m\in\Zk}\alphmk^2(1-a_m) \frac{\gamumk(\betamku-\gamumk)}{\Ntx-|\Wm|}+\nonumber
 \\
 &\rho_\UT\sum\limits_{m \in \Tk} \alphmk^2(1\!-\!a_m)\Ntx \betamku \gamumk. 
\end{align}
%=====
Similarly, we compute  $\text{UI}_{\ell k} ^\PZF$ as
%===============
\vspace{-0.5em}
\begin{align}\label{eq:UI_zf}
\mathbb{E}\big\{\vert \text{UI}_{\ell k} ^\PZF\vert^2\}&=  \rho_\UT\sum_{m\in\Zk}\alphmk^2(1-a_m) \frac{\gamumk(\betamlu-\gamuml)}{\Ntx-|\Wm|}\nonumber\\
 &+\rho_\UT\sum\limits_{m \in \Tk} \alphmk^2(1\!-\!a_m)\Ntx \betamlu \gamumk. 
\end{align}
%=====
It can be shown that  ${(\vmkz)^\dag}{\qF_{mi}}{({\mathbf{\hat g}}_{i\ell}^\dl)^{\ast}}$ is a zero-mean RV with variance $\frac{\Ntx\gamumk\beta_{mi}\gamma_{i\ell}^\dl}{\Ntx-|\Wm|}$. 
Moreover,  ${(\vmkm)^\dag}{\qF_{mi}}{({\mathbf{\hat g}}_{i\ell}^\dl)^{\ast}}$ is a zero-mean RV with variance ${\Ntx^2\gamumk\beta_{mi}\gamma_{i\ell}^\dl}$. Therefore,  we can formulate   $\mathbb{E}\big\{\vert\text{MI}_{\ell k} ^\PZF\vert^2\}$ as
%===============
\vspace{-0.5em}
\begin{align}\label{eq:MI_zf}
&\mathbb{E}\big\{\vert\text{MI}_{\ell k} ^\PZF\vert^2\}  \nonumber\\
&= {\rho _\dl} \sum\limits_{i \in {\mathcal{M}}}a_{i}\theta_{i\ell} \Big(\sum\limits_{m \in \Zk} \alphmk^2 (1-a_m)  \mathbb{E}   \{ |{(\vmkz)^\dag}{{\mathbf{F}}_{mi}}{({\mathbf{\hat g}}_{i\ell}^\dl)^{\ast}}{|^2}\} \nonumber
\\
&\qquad+\sum\limits_{m \in \Tk} \alphmk^2 (1-a_m)  \mathbb{E}   \{ |{(\vmkm)^\dag}{{\mathbf{F}}_{mi}}{({\mathbf{\hat g}}_{i\ell}^\dl)^{\ast}}{|^2}\}\Big)\nonumber\\
&={\rho _\dl} \Ntx  {\sum\limits_{i \in {\mathcal{M}}}}
a_i{\theta_{i\ell}}\gamma_{i\ell}^\dl  \big(\sum\limits_{\substack{m\in\Zk}}
		\alphmk^2(1-a_m) \frac{\gamumk \beta_{mi}}{\Ntx-|\Wm|}\nonumber
  \\
  &\qquad+\Ntx
		\!\sum\limits_{\substack{m\in\Tk}}
		\alphmk^2(1-a_m)  {\gamumk \beta_{mi}}\big).
\end{align}
%==============
Finally, by exploiting the fact that the noise and the channel estimate are
independent, $\mathbb{E}\big\{\vert\text{AN}_k ^\PZF\vert^2\}$ can be written as 
%=====
\begin{align}\label{eq:AN_zf}
&\mathbb{E}\big\{\vert\text{AN}_k ^\PZF\vert^2\}= \mathbb{E}\Bigg\{ \Big \vert\sum\limits_{m \in \Zk} \alphmk(1\!-\!a_m)  {(\vmkz)^\dag}\qw_m^\ul\Big|^2\Bigg\}\nonumber
\\
&\qquad+\mathbb{E}\Bigg\{ \Big \vert\sum\limits_{m \in \Tk} \alphmk(1\!-\!a_m)  {(\vmkz)^\dag}\qw_m^\ul\Big|^2\Bigg\} \nonumber\\
&=	\sum\limits_{\substack{m\in\Zk}}\!
		\alphmk^2(1\!\!-\!a_m)\frac{\gamumk}{\Ntx-|\Wm|}+\Ntx\!\sum\limits_{\substack{m\in\Tk}}\!\!
		\alphmk^2(1\!\!-\!a_m)\gamumk.
\end{align}
%=====
 The substitution of~\eqref{eq:DS_zf},~\eqref{eq:BI_zf},~\eqref{eq:UI_zf},~\eqref{eq:MI_zf}, and~\eqref{eq:AN_zf} into~\eqref{eq:SINRST}  yields~\eqref{eq:SINRMA_ZF1}.

%=====================================================================
%=====================================================================
%%%%%%$$$$$$$$$$$$$$$$$$$$$$$$$$$$$$$$$$$$$$$$$$$$$$$$$$$$$$$$$$$$$$$$$$$$$$$$$$$$
\balance
\bibliographystyle{IEEEtran}
\bibliography{IEEEabrv,Journal_main}

\begin{IEEEbiography}[{\includegraphics[width=1in,height=1.25in,clip,keepaspectratio]{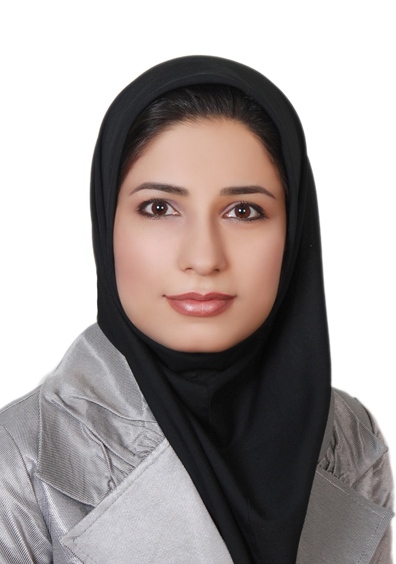}}]{Zahra Mobini}
 received the B.S. degree in electrical engineering from Isfahan University of Technology, Isfahan, Iran, in 2006, and the M.S and Ph.D.
degrees, both in electrical engineering, from the M. A. University of Technology and K. N. Toosi University of Technology, Tehran, Iran, respectively. From November 2010 to November 2011, she was a Visiting Researcher at the Research School of Engineering, Australian National University, Canberra, ACT, Australia. She is currently  a Post-Doctoral Research Fellow at the Centre for Wireless Innovation (CWI), Queen's University Belfast (QUB). Before joining QUB,  she was an Assistant and then Associate Professor with the Faculty of Engineering, Shahrekord University, Shahrekord, Iran (2015-2021). 
Her research interests include physical-layer security, massive  MIMO, cell-free massive  MIMO, full-duplex communications, and resource management and optimization. She has co-authored many research papers in wireless communications. She has actively served as the reviewer for a variety of IEEE journals,  such as TWC, TCOM, and TVT.

\end{IEEEbiography}

\begin{IEEEbiography}[{\includegraphics[width=1in,height=1.25in,clip,keepaspectratio]{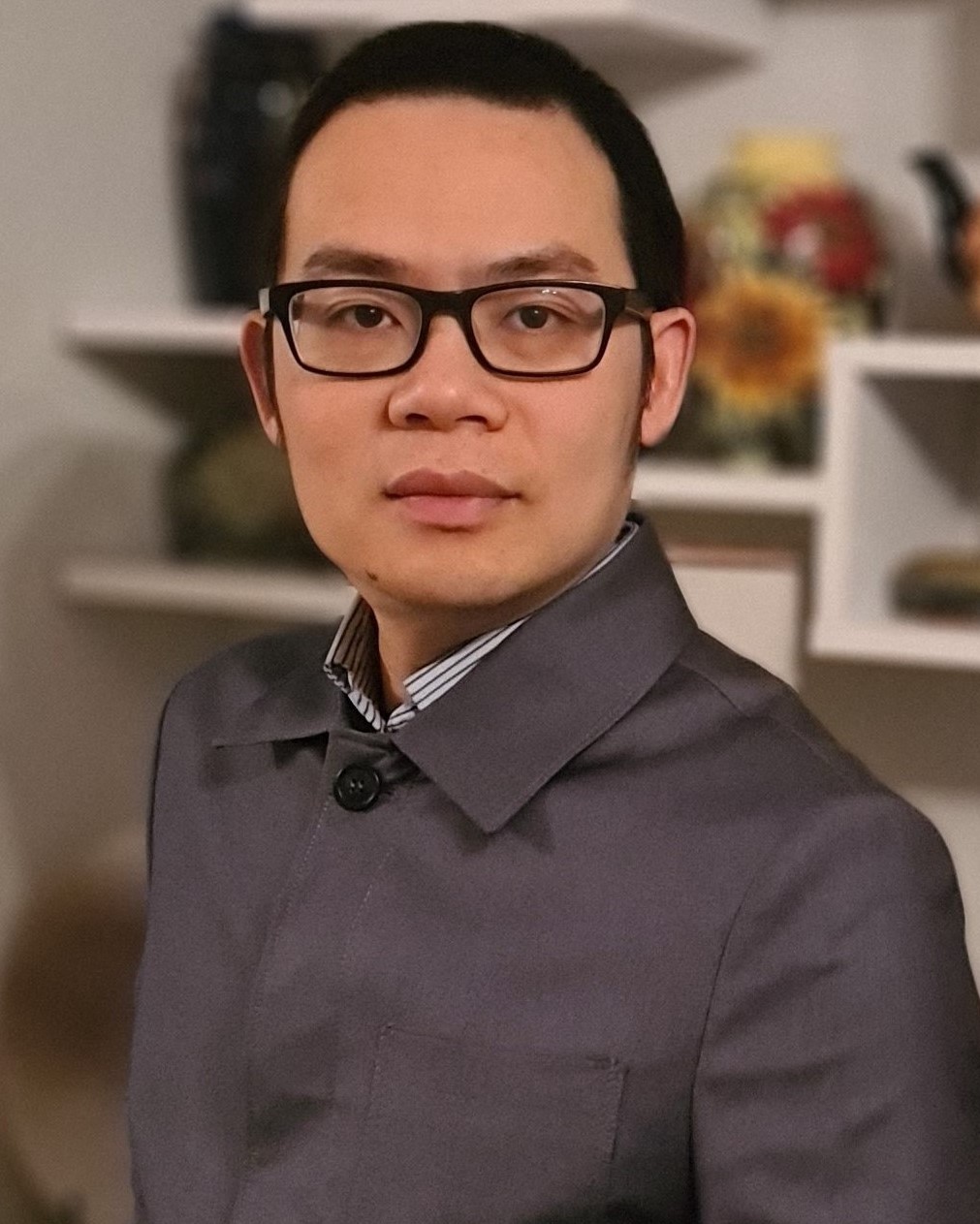}}]
{Hien Quoc Ngo} is currently a Reader with Queen's University Belfast, U.K. His main research interests include massive MIMO systems, cell-free massive MIMO, reconfigurable intelligent surfaces, physical layer security, and cooperative communications. He has co-authored many research papers in wireless communications and co-authored the Cambridge University Press textbook \emph{Fundamentals of Massive MIMO} (2016).

He received the IEEE ComSoc Stephen O. Rice Prize in 2015, the IEEE ComSoc Leonard G. Abraham Prize in 2017, the Best Ph.D. Award from EURASIP in 2018, and the IEEE CTTC Early Achievement Award in 2023. He also received the IEEE Sweden VT-COM-IT Joint Chapter Best Student Journal Paper Award in 2015. He was awarded the UKRI Future Leaders Fellowship in 2019. He serves as the Editor for the IEEE Transactions on Wireless Communications, IEEE Transactions on Communications, the Digital Signal Processing, and the Physical Communication (Elsevier). He was a Guest Editor of IET Communications, and a Guest Editor of IEEE ACCESS in 2017.
\end{IEEEbiography}

\begin{IEEEbiography}[{\includegraphics[width=1in,height=1.25in,clip,keepaspectratio]{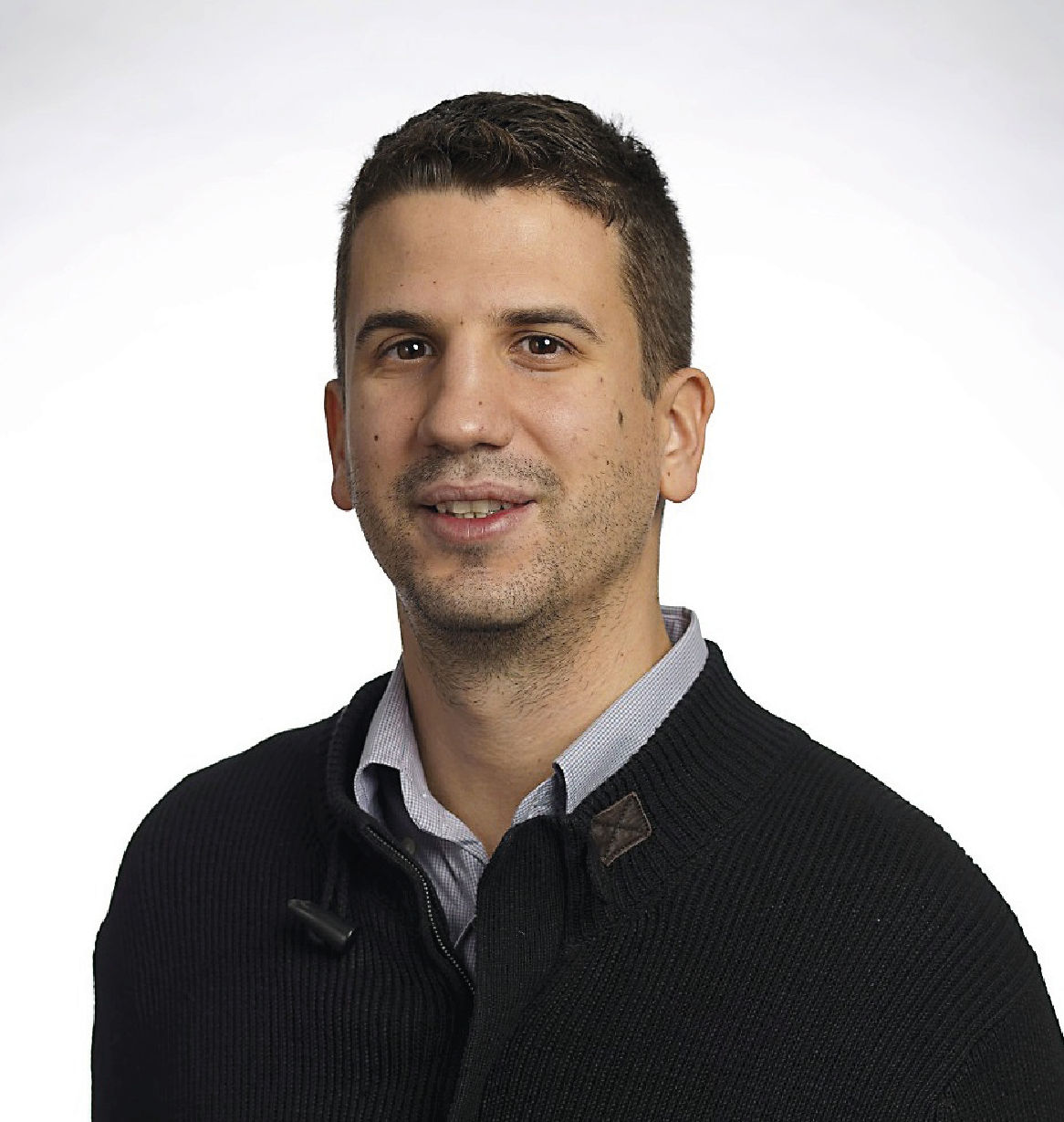}}]
{Michail Matthaiou}(Fellow, IEEE) was born in Thessaloniki, Greece in 1981. He obtained the Diploma degree (5 years) in Electrical and Computer Engineering from the Aristotle University of Thessaloniki, Greece in 2004. He then received the M.Sc. (with distinction) in Communication Systems and Signal Processing from the University of Bristol, U.K. and Ph.D. degrees from the University of Edinburgh, U.K. in 2005 and 2008, respectively. From September 2008 through May 2010, he was with the Institute for Circuit Theory and Signal Processing, Munich University of Technology (TUM), Germany working as a Postdoctoral Research Associate. He is currently a Professor of Communications Engineering and Signal Processing and Deputy Director of the Centre for Wireless Innovation (CWI) at Queen’s University Belfast, U.K. after holding an Assistant Professor position at Chalmers University of Technology, Sweden. His research interests span signal processing for wireless communications, beyond massive MIMO, intelligent reflecting surfaces, mm-wave/THz systems and deep learning for communications.

Dr. Matthaiou and his coauthors received the IEEE Communications Society (ComSoc) Leonard G. Abraham Prize in 2017. He currently holds the ERC Consolidator Grant BEATRICE (2021-2026) focused on the interface between information and electromagnetic theories. To date, he has received the prestigious 2023 Argo Network Innovation Award, the 2019 EURASIP Early Career Award and the 2018/2019 Royal Academy of Engineering/The Leverhulme Trust Senior Research Fellowship. His team was also the Grand Winner of the 2019 Mobile World Congress Challenge. He was the recipient of the 2011 IEEE ComSoc Best Young Researcher Award for the Europe, Middle East and Africa Region and a co-recipient of the 2006 IEEE Communications Chapter Project Prize for the best M.Sc. dissertation in the area of communications. He has co-authored papers that received best paper awards at the 2018 IEEE WCSP and 2014 IEEE ICC. In 2014, he received the Research Fund for International Young Scientists from the National Natural Science Foundation of China. He is currently the Editor-in-Chief of Elsevier Physical Communication, a Senior Editor for \textsc{IEEE Wireless Communications Letters} and \textsc{IEEE Signal Processing Magazine}, and an Area Editor for \textsc{IEEE Transactions on Communications}. He is an IEEE and AAIA Fellow.
\end{IEEEbiography}

\begin{IEEEbiography}[{\includegraphics[width=1in,height=1.25in,clip,keepaspectratio]{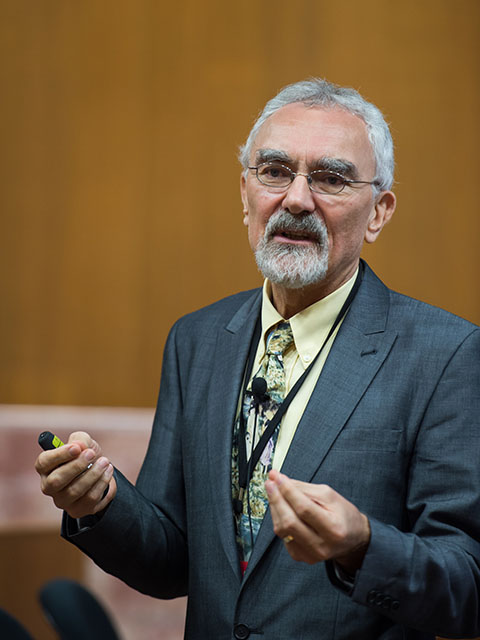}}]
{Lajos Hanzo} (FIEEE'04) received Honorary Doctorates  from the Technical University of Budapest (2009) and Edinburgh University (2015). He is a Foreign Member of the Hungarian Science-Academy, Fellow of the Royal Academy of Engineering (FREng), of the IET, of EURASIP and holds the IEEE Eric Sumner Technical Field Award. For further details please see \url{http://www-mobile.ecs.soton.ac.uk}, \url{https://en.wikipedia.org/wiki/Lajos_Hanzo}.

\end{IEEEbiography}

\end{document}